%% file: clustering_arxiv.tex
\pdfoutput=1
\documentclass[10pt,twocolumn]{article}
\usepackage{amsmath,amssymb}
\usepackage{epsfig}
\usepackage{graphicx}
\usepackage{verbatim}
\usepackage{tipa}
\usepackage{algorithmic}
\usepackage{bm}
\usepackage{flushend}
\usepackage{subfigure}
\usepackage{url}
\usepackage{verbatim}
\usepackage[suffix=,update]{epstopdf}

\newenvironment{proof}{\paragraph{Proof.}}{\hfill$\square$}

\topmargin by -\headsep \textheight 8.9in \oddsidemargin 0pt
\evensidemargin \oddsidemargin
\def\cl{\mbox{\it cl}\kern.2ex}
\def\C{{\mathcal C}\kern.1ex}
\def\,{\kern.15ex}

\newtheorem{theorem}{Theorem}[section]
\newtheorem{lemma}[theorem]{Lemma}

\begin{document}

\title{Scalable Constrained Clustering: \\
A Generalized Spectral Method \thanks{Significant part of this work
was carried out while M. Cucuringu, I. Koutis and G. Miller were visiting the
Simons Institute for the Theory of Computing at UC Berkeley in Fall 2014. I. Koutis is supported
by NSF CAREER award CCF-1149048.}}

\author{
Mihai Cucuringu  \\U. of California, Los Angeles \\ mihai@math.ucla.edu
\and Ioannis Koutis \\ U. of Puerto Rico - Rio Piedras\\ ioannis.koutis@upr.edu
\and Sanjay Chawla \\  Qatar C  , HBKU\footnote{and University of Sydney.} \\ chawla@it.usyd.edu.au
\and Gary Miller \\ Carnegie Mellon University \\ glmiller@cs.cmu.edu
\and Richard Peng \\ Georgia Institute of Technology \\ rpeng@cc.gatech.edu
}

\maketitle

\begin{abstract}
We present a simple spectral approach to the well-studied constrained
clustering problem. It captures constrained clustering as a generalized eigenvalue problem with graph Laplacians. The algorithm works in nearly-linear time and provides
concrete guarantees for the quality of the clusters, at least for the
case of 2-way partitioning. In practice this translates
to a very fast implementation that consistently outperforms
existing spectral approaches both in speed and quality.
\end{abstract}

\input{intro}

\input{probdef}

\input{related}

\input{algo}

\input{cheeger}

\input{experiments}

\input{conclusions}

{\small
\flushend
\bibliographystyle{nalpha}

\input{clustering_arxiv.bbl}
}

\newpage
\onecolumn
\input{supplementary-file}



\end{document}

%% file: intro.tex
\section{Introduction}\label{section:intro}
Clustering with constraints is a problem of central importance
in machine learning and data mining. It captures the case when information about an
application task comes in the form of both data and domain knowledge.
We study the standard problem where domain knowledge is
specified as a set of {\em soft} must-link (ML) and cannot-link (CL) constraints
\cite{basuDavidson}.

The extensive literature reports a plethora of methods, including
spectral algorithms that explore various modifications and extensions of the basic
spectral algorithm by Shi and Malik~\cite{ShiM00} and its
variant by Ng et al.~\cite{NgJW01}.

The distinctive feature of our algorithm is that it constitutes a natural \textbf{generalization}, rather
than an extension of the basic spectral method. The generalization is based on a critical look at how existing methods handle constraints, in section~\ref{sec:rethinking}.
The solution is derived from a geometric embedding obtained via a spectral relaxation of an
optimization problem, exactly in the spirit of~\cite{NgJW01,ShiM00}. This is depicted
in the workflow in Figure~1.  Data
and ML constraints are represented by a Laplacian matrix $L$ and CL constraints by another
Laplacian matrix $H$. The embedding is realized by computing a few eigenvectors
of the generalized eigenvalue problem $Lx = \lambda Hx$. The generalization of~\cite{NgJW01,ShiM00}
lies essentially in $H$ being a Laplacian matrix rather than the diagonal $D$ of $L$. In
fact, as we will discuss later, $D$ itself is equivalent to a specific Laplacian matrix;
thus our method encompasses the basic spectral method as a special case of constrained clustering.

\begin{figure}[h]
\includegraphics[width=1\columnwidth]{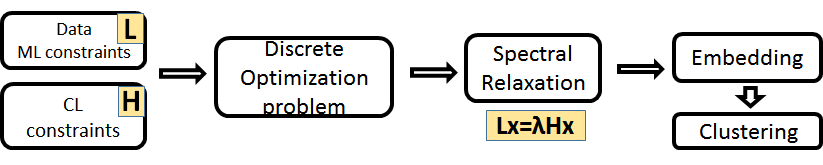}
\caption{A schematic overview of our approach.}
\label{fig:flow_gen}
\end{figure}

Our approach is characterized by its conceptual simplicity that enables
a straightforward mathematical derivation of the algorithm, possibly the simplest among all
competing spectral methods. Reducing the problem to a relatively simple generalized
eigensystem enables us to derive directly from recent significant progress  due to Lee~et~al.~\cite{Lee12} in the
theoretical understanding of the standard spectral clustering method,
offering its first practical realization. In addition, the algorithm comes with
two features that are not simultaneously shared by {\em any} of the prior methods:
(i) it is provably fast by design as it leverages
fast linear system solvers for Laplacian systems~\cite{Koutis:2012}
(ii) it provides a concrete theoretical
guarantee for the quality of 2-way constrained partitioning, with
respect to the underlying discrete optimization problem,
via a generalized Cheeger inequality (section~\ref{sec:cheeger}).

In practice, our method is at least 10x faster than
competing methods on large data sets. It solves data sets with millions of
points in less than 2 minutes, on very modest hardware. Furthermore
the quality of the computed segmentations is often dramatically better.




%% file: probdef.tex
\section{Problem definition}
\label{sec:pdef}

The constrained clustering problem is specified by three weighted graphs:

\smallskip
\noindent \textbf{1.} The {\em data graph} $G_D$ which contains a given number of $k$ clusters that we seek to find. Formally, the graph is a triple $G_D=(V,E_D,w_D)$, with the edge weights $w_D$ being positive real numbers indicating the level of `affinity' of their endpoints.

\smallskip

\noindent \textbf{2.} The {\em knowledge graphs} $G_{ML}$ and $G_{CL}$.
The two graphs are formally triples $G_{ML}=(V,E_{ML},w_{ML})$
and $G_{CL}=(V,E_{CL},w_{CL})$.
Each edge in $G_{ML}$ indicates that its two
endpoints should be in the same cluster, and each
edge in $G_{CL}$ indicates that its two endpoints should
be in different clusters. The weight of an edge indicates the level of belief placed in the corresponding constraint.

\smallskip

We emphasize that prior knowledge does not have to be exact
or even self-consistent, and thus the constraints should not
be viewed as `hard' ones. However, to conform with prior
literature, we will use the existing terminology of `must
link' (ML) and `cannot link' (CL) constraints to which
 $G_{ML}$ and $G_{CL}$ owe their notation respectively.

In the constrained clustering problem the general goal is to find $k$ {disjoint}  clusters in the data graph. Intuitively, the clusters should result from cutting a small number of edges in the data graph, while simultaneously respecting as much as possible the constraints in the knowledge graphs.


\section{Re-thinking constraints}
\label{sec:rethinking}

Many approaches have been pursued within the constrained spectral clustering framework. They are quite distinct but do share a common point of view: constraints are viewed as entities structurally extraneous to the basic spectral formulation, necessitating its modification or extension with additional mathematical features. However, a key fact is overlooked:

\begin{center}
{\em Standard clustering is a special case of constrained clustering with implicit soft ML and CL constraints}.
\end{center}

To see why, let us briefly recall the optimization problem in the standard method (\textsc{Ncut}).
$$
     \phi =  \min_{S\subseteq V} \frac{cut_{G_D}(S,\bar{S}) }{vol(S)vol(\bar{S})/vol(V)}.
$$
Here $vol(S)$ denotes the total weight incident to the vertex set $S$, and $cut_G(S,\bar{S})$ denotes the total weight crossing from $S$ to $\bar{S}$ in $G$.

The data graph $G_D$ is actually an implicit encoding of soft ML constraints. Indeed, pairwise affinities between nodes can be viewed as `soft declarations' that such nodes should be connected rather than disconnected in a clustering. Let now $d_i$ denote the total incident weight of vertex $i$ in $G_D$. Consider the {\bf demand graph}~$K$ of implicit soft CL constraints, defined by the adjacency $K_{ij}=d_id_j/vol(V)$. It is easy to verify that $vol(S)vol(\bar{S})/vol(V)=cut_K(S,\bar{S})$. We have
$$
  \min_{S\subseteq V} \frac{cut_{G_D}(S,\bar{S}) }{vol(S)vol(\bar{S})/vol(V)} =  \min_{S\subseteq V}\frac{cut_{G_D}(S,\bar{S}) }{cut_{K}(S,\bar{S}) }.
$$
In other words, the \textsc{Ncut} objective can be viewed as:
\begin{equation} \label{eq:rethink}
   \min_{S\subseteq V} \frac{\textnormal{weight of cut (violated) implicit ML constraints}}{\textnormal{weight of cut (satisfied) implicit CL constraints}}.
\end{equation}

With this realization, it becomes evident that incorporating the knowledge graphs ($G_{ML},G_{CL}$) is mainly a degree-of-belief issue, between implicit and {\em explicit constraints}.  Yet all existing methods insist on handling the explicit constraints separately. For example, \cite{RangapuramH12} modify the \textsc{Ncut} optimization function by adding in the numerator the number of violated explicit constraints (independently of them being ML or CL), times a parameter $\gamma$. In another example, \cite{WangQD14} solve the spectral relaxation of \textsc{Ncut}, but under the constraint that the number of satisfied ML constraints minus the number of violated CL constraints is lower bounded by a parameter $\alpha$. Despite the separate handling of the explicit constraints, degree-of-belief decisions  (reflected by parameters $\alpha$ and $\gamma$) are not avoided. The actual handling also appears to be somewhat arbitrary. For instance, most methods take the constraints unweighted, as usually provided by a user, and handle them uniformly; but it is unclear why one constraint in a densely connected part of the graph should be treated equally to another constraint in a less well-connected part. Moreover, most prior methods enforce the use of the balance implicit constraints in $K$, without questioning their role, which may be actually adverserial in some cases.
In general, the mechanisms for including the explicit constraints are {\em oblivious} of the input, or even of the underlying algebra.

\noindent \textbf{\large Our approach.} We choose to temporarily drop the distinction of the constraints into explicit and implicit. We instead assume that we are given one set of ML constraints, and one set of CL constraints, in the form of weighted graphs $G$ and $H$. We then design a generalized spectral clustering method that retains the $k$-way version of the objective shown in equation~\ref{eq:rethink}. We apply this generalized method to our original problem, after a {\em merging step} of the explicit and implicit CL/ML constraints into one set of CL/ML constraints.

The merging step can be left entirely up to the user, who may be able to exploit problem-specific information and provide their choice of weights for $G$ and $H$. Of course, we expect that in most cases explicit CL and ML constraints will be provided in the form of simple unweighted graphs $G_{ML}$ and $G_{CL}$. For this case we provide a simple method that resolves the degree-of-belief issue and constructs $G$ and $H$ {\em automatically}. The method is heuristic, but not oblivious to the data graph, as they adjust to it. 


%% file: related.tex
\section{Related Work}
\label{section:related}

The literature on constrained clustering is quite extensive, as
the problem has been pursued under various guises from different communities.
Here we present a short and unavoidably partial review.   

A number of methods incorporate the constraints via only
modifying the data matrix in the standard method. In certain 
cases some or all of the CL constraints are dropped in order to prevent
the matrix from turning negative~\cite{KamvarKM03,DBLP:conf/cvpr/LuC08}.
The formulation of~\cite{RangapuramH12} incorporates all constraints 
into the data matrix, essentially by adding a {\em signed Laplacian},
which is a generalization of the Laplacian for graphs with negative weights;
notably, their algorithm does not solve a spectral relaxation of the problem but
attempts to solve the (hard) optimization problem exactly, via a continuous
optimization approach. 
 
A different approach is proposed in~\cite{LiLT09}: constraints are used 
in order to improve the embedding obtained through the standard problem, 
before applying the partitioning step. In principle this embedding-processing
step is orthogonal to methods that compute some embedding (including ours), 
and it can be used to potentially improve them.

A number of other works
use the ML and CL constraints to super-impose algebraic constraints onto the spectral  relaxation
of the standard problem.
These additional algebraic constraints usually yield
much harder constrained optimization problems~\cite{ErikssonOK11,BoleyK13,XuLS09,WangQD14}.

Besides our work, there exists a number of other approaches that reduce constrained
clustering into generalized eigenvalue problems $Ax =\lambda B x$ that 
deviate substantially from than the standard formulation. 
These methods can be implemented to run fast, as long as:
(i) linear systems in $A$ can be solved efficiently,
(ii) $A$ and $B$ are positive semi-definite.  
Specifically, 
\cite{DBLP:conf/cvpr/YuS01, DBLP:journals/pami/YuS04} use
a generalized eigenvalue problem in which $B$ is 
a diagonal, but $A$ is not generally amenable to existing efficient
linear system solvers.
In~\cite{WangQD14} matrix $A$ is set to be 
 the normalized Laplacian of the data graph
(implicitly attempting to impose the standard balance constraints),
and $B$ has both positive and negative off-diagonal
entries representing ML and CL constraints respectively. In the general case $B$ is not
positive, forcing the computation of full eigenvalue
decompositions. However the method can be modified to use a (positive)
signed Laplacian as the matrix $B$, as partially observed in~\cite{Wang:2012}.
This modification has a fast implementation. The formulation in~\cite{RangapuramH12} 
also leads to a fast implementation of its spectral relaxation.

%% file: algo.tex
\section{Algorithm and its derivation} \label{sec:algorithms}

\subsection{Graph Laplacians} \label{sec:lap}
Let $G=(V,E,w)$ be a graph with positive weights. The {\em Laplacian} $L_G$ of $G$
is defined by $L_G(i,j) = -w_{ij}$ and
$L_G(i,i) = \sum_{j\neq i} w_{ij}$.
The graph Laplacian satisfies the following basic identity for all
vectors $x$:
\begin{eqnarray} \label{eq:quadratic}
    x^T L_G x & =& \sum_{i,j} w_{ij}(x_i-x_j)^2.
\end{eqnarray}
Given a cluster $C\subseteq V$ we define a cluster indicator vector by
$ x_C(i) = 1$ if $i \in C$ and $x_C(i)=0$ otherwise. We have:
\begin{eqnarray} \label{eq:LaplacianCut}
   x_C^T L_G x_C = cut_G(C,\bar{C})
\end{eqnarray}
where $cut_G(C,\bar{C})$ denotes the total weight crossing from $C$ to $\bar{C}$ in $G$.

\subsection{The optimization problem}

As we discussed in section \ref{sec:rethinking}, we assume that the input consists of two
weighted graphs, the must-link constraints $G$, and the cannot-link
constraints $H$.

Our objective is to partition the node set $V$ into k disjoint clusters $C_{i}$.
We define an individual measure of {\em badness}
for each cluster $C_i$:
\begin{equation}
\phi_i(G,H) = \frac{cut_{G}(C_i,\bar{C_i})}{cut_{H}(C_i,\bar{C_i})}
\end{equation}
The numerator is equal to the total weight of the violated ML constraints, because cutting
one such constraint violates it. The denominator is equal to the total weight of the satisfied
CL constraints, because cutting one such constraint satisfies it. Thus the minimization
of the individual badness is a sensible objective.

We would like then to find
clusters $C_1,\ldots,C_k$ that minimize the maximum badness, i.e. solve the following problem:
\begin{equation} \label{eq:clustergoal}
{\Phi}_k = \textnormal{{\bf min} $\max_{i} {\phi_i}$}.
\end{equation}

Using equation \ref{eq:LaplacianCut}, the above  can be captured in terms of Laplacians: letting $x_{C_i}$
denote the indicator vector for cluster $i$, we have
$$
   {\phi}_i(G,H) = \frac{x_{C_i}^T L_{G} x_{C_i}}{x_{C_i}^TL_{H}x_{C_i}}.
$$
Therefore, solving the minimization problem posed in equation \ref{eq:clustergoal}
amounts to finding $k$ vectors in $\{0,1\}^n$ with disjoint support.

Notice that the optimization problem may not be well-defined in the event that there are very few
CL constraints in $H$. This can be detected easily and the user can be notified. The merging phase
also takes automatically care of this case. Thus we assume that the problem is well-defined.

\subsection{Spectral Relaxation}

To relax the problem we instead look for $k$ vectors in $y_1,
\ldots, y_k \in {\mathbb R}^n$, such that for all $i\neq j$, we have
$y_i L_H y_j = 0$. These $L_H$-{orthogonality} constraints can be viewed
as a relaxation of the disjointness requirement. Of course their
particular form is motivated by the fact that they directly give rise
to a generalized eigenvalue problem. Concretely, the $k$ vectors $y_i$
that minimize the maximum among the $k$
Rayleigh quotients $(y_i^T L_G y_i)/(y_i^T L_H y_i)$ are precisely
the generalized eigenvectors corresponding to the $k$ smallest
eigenvalues of the problem:
$
     L_G x = \lambda L_H x.
$\footnote{When $H$ is the demand graph $K$ discussed in section~\ref{sec:pdef},
the problem is identical to the standard problem $L_Gx = \lambda D x$, where $D$ is the diagonal
of $L_G$. This is because $L_K = D - dd^T/(d^T {\bf 1}$), and the eigenvectors
of $L_Gx = \lambda D x$ are $d$-orthogonal, where $d$ is vector of degrees in $G$.}
This fact is well understood and follows from a generalization
of the min-max characterization of the eigenvalues for symmetric
matrices; details can be found for instance in~\cite{Stewart.Sun}.

Notice that $H$ does not have to be connected. Since we are looking
for a minimum, the optimization function
avoids vectors that are in the null space of $L_H$. That means that no restriction
needs to be placed on $x$ so that the eigenvalue problem is well defined, other than
it can't be the constant vector (which is in the null space of both $L_G$ and $L_H$),
assuming without loss of generality that $G$ is connected.


\subsection{The embedding}

Let $X$ be the $n\times k$ matrix of the first $k$ generalized eigenvectors
for $L_G x= \lambda L_H x$. The embedding is shown in Figure~\ref{EmbeddingSpace}.

\renewcommand{\algorithmicrequire}{\textbf{Input:}}
\renewcommand{\algorithmicensure}{\textbf{Output:}}

We discuss the intuition behind the embedding.
Without step~4 and with $L_H$ replaced with the diagonal $D$,
the embedding is exactly the one recently proposed and analyzed in~\cite{Lee12}.
It is a combination of the embeddings
considered in~\cite{ShiM00,NgJW01,Verma03acomparison}, but
the first known to produce clusters with approximation
guarantees.
The generalized eigenvalue problem $Lx = \lambda D x$ can be
viewed as a simple eigenvalue problem over a space
endowed with the $D$-inner product: $\left<x,y\right>_D = x^T D y$.
Step~5 normalizes the eigenvectors
to a unit $D$-norm, i.e. $x^T D x=1$. Given this normalization,
it is shown in~\cite{Lee12} that the rows of $U$ at step~7 (vectors
in $k$-dimensional space) are expected to concentrate in $k$ different {\em directions}.
This justifies steps~8-10 that normalize these row vectors onto the $k$-dimensional
sphere, in order to concentrate them in a {\em spatial} sense. Then a geometric
partitioning algorithm can be applied.

\begin{figure} [h]
\begin{algorithmic}[1]
\REQUIRE
 $X,L_H,d$ \\
\ENSURE embedding $U \in {\mathbb R}^{n\times k}$, $l \in {\mathbb R}^{n\times 1}$
\STATE $u \leftarrow  1^n$
\FOR {$i=1:k$}
\STATE $x = X_{:,i}$
\STATE $x = x - (x^Td/u^Td)u$
\STATE $x = x/\sqrt{x^T L_H x}$
\STATE $U_{:,i} = x$
\ENDFOR
\FOR {$j=1:n$}
\STATE $l_j = ||U_{j,:}||_2$
\STATE $U_{j,:} = U_{j,:}/l_j$
\ENDFOR
\end{algorithmic}
\caption{Embedding Computation (based on~\cite{Lee12}).}
\label{EmbeddingSpace}
\end{figure}

From a technical point of view, working with $L_H$ instead of $D$
makes almost no difference. $L_H$ is a positive definite matrix. It
can be rank-deficient, but the eigenvectors avoid the null space
of $L_H$, by definition. Thus the geometric intuition about $U$ remains
the same if we syntactically replace $D$ by $L_H$.
However, there is a subtlety: $L_G$ and $L_H$ share the constant
vector in their null spaces. This means that if $x$ is an eigenvector,
then for all $c$ the vector $x+ c{\bf 1}^n$ is also an eigenvector
with the same eigenvalue. Among all such possible eigenvectors we pick
one representative: in Step~4 we pick $c$ such that $x+ c{\bf 1}^n$ is orthogonal to~$d$.
The intuition for this is derived from the proof of the Cheeger inequality
claimed in section~\ref{sec:cheeger}; this choice is what
makes possible the analysis of a theoretical guarantee for a 2-way cut.


\subsection{Computing Eigenvectors}

It is understood that spectral algorithms based on eigenvector embeddings
do not require the exact eigenvectors, but only approximations of them,
in the sense that the quotients $x^T L x/x^T H x$ are close to
their exact values, i.e. close to the eigenvalues~\cite{chung1,Lee12}.
The computation of such approximate  generalized eigenvectors for $L_G x = \lambda L_H x$
is the most time-consuming part of the entire process. The asymptotically fastest known
algorithm for the problem runs in
$O(km \log^2 m)$ time. It combines a fast Laplacian linear system solver~\cite{KoutisMP11}
and a standard power method~\cite{GoVa96}.
In practice we use the combinatorial multigrid solver~\cite{KoutisMT11} which empirically
runs in $O(m)$ time. The solver provides an approximate inverse for $L_G$
which in turn is used with the preconditioned eigenvalue solver \textsc{LOBPCG}~\cite{Knyazev01towardthe}.

\subsection{Partitioning}

For the special case when $k=2$, we can compute the second
eigenvector, sort it, and then select the sparsest cut among
the $n-1$ possible cuts into
$\{v_1,\ldots,v_i\}$ and $\{v_{i+1}\ldots v_n\}$, for $i\in [1,n]$,
where $v_j$ is the vertex that corresponds to coordinate $j$ after
the sorting. This `Cheeger sweep' method is associated with the proof of the
Cheeger inequality~\cite{chung1}, and is also used in the proof
of the inequality we claim in section~\ref{sec:cheeger}.



In the general case, given the embedding matrix embedding $U$, the clustering algorithm
invokes \texttt{kmeans}$(U)$ (with a random start), which returns a $k$-partitioning.
The partitioning can be  refined optionally into a $k$-clustering
by performing a Cheeger sweep among the nodes of each component,
independently for each component:
the nodes are sorted  according to the values of the corresponding coordinates in the vector $l$
returned by the embedding algorithm given in~\ref{EmbeddingSpace}. We will not
use this refinement option in our experiments.

\subsection{Merging Constraints}

As we discussed in section~\ref{sec:pdef}, it is frequently
the case that a user provides unweighted constraints
$G_{ML}$ and $G_{CL}$. Merging these unweighted constraints
with the data into one pair of graphs $G$ and $H$
is an interesting problem.

Here we propose a simple heuristic. We construct
two weighted graphs $\hat{G}_{ML}$ and $\hat{G}_{CL}$,
as follows: if edge $(i,j)$ is a constraint,
we take its weight in the corresponding
graph to be $d_id_j/(d_{\min}d_{\max})$, where $d_i$
denotes the total incident weight of vertex $i$,
and $d_{\min},d_{\max}$ the minimum and maximum
among the $d_i$'s.
We then let $G= G_D + \hat{G}_{ML}$ and $H = K/n+ \hat{G}_{CL}$,
where $K$ is the demand graph
and $n$ is the size of the data graph, whose edges are
normalized to have minimum weight. We include this small
copy of $K$ in $H$ in order to render the problem well-defined in all cases
of user input.

The intuition behind this choice of weights is better understood in the
context of a sparse unweighted graph. A constraint
on two high-degree vertices is more significant relative to
a constraint on two lower-degree vertices, as it has the potential to
drastically change the clustering, if enforced.
In addition,
assuming that noisy/inaccurate constraints are uniformly random,
there is a lower probability
that a high-degree constraint is inaccurate, simply because its two
endpoints are relatively rare, due to their high degree.
From an algebraic point of view, it also makes sense
having a higher weight on this edge, in order to be
comparable with the neighborhood of $i$ and $j$ and
have an effect in the value of the objective function. Notice
also that when no constraints are available the method reverts
to standard spectral clustering.

%% file: cheeger.tex
\section{A generalized Cheeger inequality} \label{sec:cheeger}

The success of the standard spectral clustering method is often
attributed to the existence of non-trivial approximation guarantees,
which in the 2-way case is given by the Cheeger inequality and the
associated method~\cite{chung1}.
Here we present a generalization of the Cheeger inequality.
We believe that it provides supporting mathematical evidence for the advantages of
expressing the constrained clustering problem as a generalized
eigenvalue problem with Laplacians.

\begin{theorem}
\label{thm:generalizedcheeger}

Let $G$ and $H$ be any two weighted graphs and $d$ be
the vector containing the degrees of the vertices in $G$.
For any vector $x$ such that $x^Td =0$, we have

\[
\frac{x^TL_Gx}{x^TL_Hx} \geq \phi(G,K) \cdot \phi(G, H)/4,
\]
where $K$ is the demand graph.
A cut meeting the guarantee of the inequality can
be obtained via a Cheeger sweep on $x$.
\end{theorem}

Due to its length, the proof is given separately in section~\ref{sec:proof}.

%% file: experiments.tex
\section{Experiments}
\label{sec:experiments}
In this section, we sample some of our experimental results. We compare our algorithm \textbf{Fast-GE} against two other methods, \textbf{CSP}~\cite{WangQD14} and \textbf{COSC}~\cite{RangapuramH12}.

\textbf{COSC} is an iterative algorithm that attempts to solve exactly an NP-hard discrete optimization problem that captures 2-way constrained clustering; $k$-way partitions are computed via recursive calls to the 2-way partitioner. The  method actually comes in two variants, an exact version which is very slow in all but very small problems, and an approximate `fast' version which has no convergence guarantees. The size of the data in our experiments forces us to use the fast version, \textbf{COSf}.

\textbf{CSP} reduces constrained clustering to a generalized eigenvalue problem. However, the problem is indefinite and the method requires the computation of a full eigenvalue decomposition.

We focus on these two methods because of their readily available implementations but mostly because the corresponding papers provide sufficient evidence that they outperform other competing methods. We also selected them because they can be both modified or extended into methods that have fast implementations.

\subsection{Some negative findings.} \textbf{COSC} has a natural spectral relaxation into a generalized eigenvalue problem $Ax = \lambda B x$ where $A$ is a signed Laplacian and $B$ is a diagonal. \textbf{CSP} can also be modified by replacing the indefinite matrix $Q$ of its generalized eigenvalue problem with a signed Laplacian that counts the number of satisfied constraints. In this way both methods become scalable. We did a number of experiments based on these observations. The results were disappointing, especially when $k>2$. The output quality was comparable or worse to that obtained by \textbf{COSf} and \textbf{CSP} in the reported experiments. We attribute this the less-clean mathematical properties of the signed Laplacian.

We also experimented with the automated merging phase of \textbf{Fast-GE}. Specifically we tried adding more significance to the standard implicit balance constraints, by increasing the coefficient of the demand graph $K$ in graph $H$. The output deteriorates (often significantly) for the more challenging problems we tried. This supports our decision to not enforce the use of balance constraints in our generalized formulation, unlike all prior methods.

\subsection{Synthetic Data Sets.} We begin with a number of small synthetic experiments. The purpose is to test the output quality, especially under the presence of noise.

We generically apply the following construction: we chose uniformly at random a set of nodes for which we assume cluster-membership information is provided. The cluster-membership information gives unweighted ML and CL constraints in the obvious way. We also add random noise in the data.

More concretely, we say that a graph $G$ is generated from the ensemble \textit{NoisyKnn}($n,k_g,l_g$)  with parameters $n$, $k_g$ and $l_g$ if $G$ of size $n$ is the union of two (non-necessarily disjoint) graphs $H_1$ and $H_2$ each on the same set of $n$ vertices
$ G = H_1 \cup H_2,$
where $H_1$ is a k-nearest-neighbor (knn) graph with each node connected to its $k_g$ nearest neighbors, and $H_2$ is an Erd\H{o}s-R\'{e}nyi graph where each edge appears independently with probability ${l_g}/{n}$. One may interpret the  parameter  $l_g$ as the noise level in the data, since the larger $l_g$ the more random edges are wired across the different clusters, thus rendering the problem more difficult to solve. In other words, the \textit{planted} clusters are harder to detect when there is a large amount of noise in the data, obscuring the separation of the clusters.

Since in these synthetic data sets, the ground truth partition is available, we measure the accuracy of the methods by the popular Rand Index~\cite{rand1971}. The Rand Index indicates how well the resulting partition matches the ground truth partition; a value closer to 1 indicates an almost perfect recovery, while a value closer to 0 indicates an almost random assignment of the nodes into clusters.

\noindent \textbf{Four Moons.}  Our first synthetic example is the `Four-Moons' data set, where the underlying graph $G$ is generated from the ensemble
\textit{NoisyKnn}($n=1500, k_g=30,l_g=15$). 
The plots in Figure~\ref{fig:FourMoons1500_curves} show the accuracy and running times of all three methods on this example, while Figure~\ref{fig:FourMoons1500_output} shows a random instance of the clustering returned by each of the methods, with 75 constraints. The accuracy of \textbf{FAST-GE} and \textbf{COSf} is very similar,
with \textbf{FAST-GE} being somewhat better with more constraints, as shown in Figure~\ref{fig:FourMoons1500_curves}. However \textbf{FAST-GE} is already at least 4x faster than \textbf{COSf}, for this size.

\begin{figure}[h]
\begin{center}
\includegraphics[width=0.32\columnwidth]{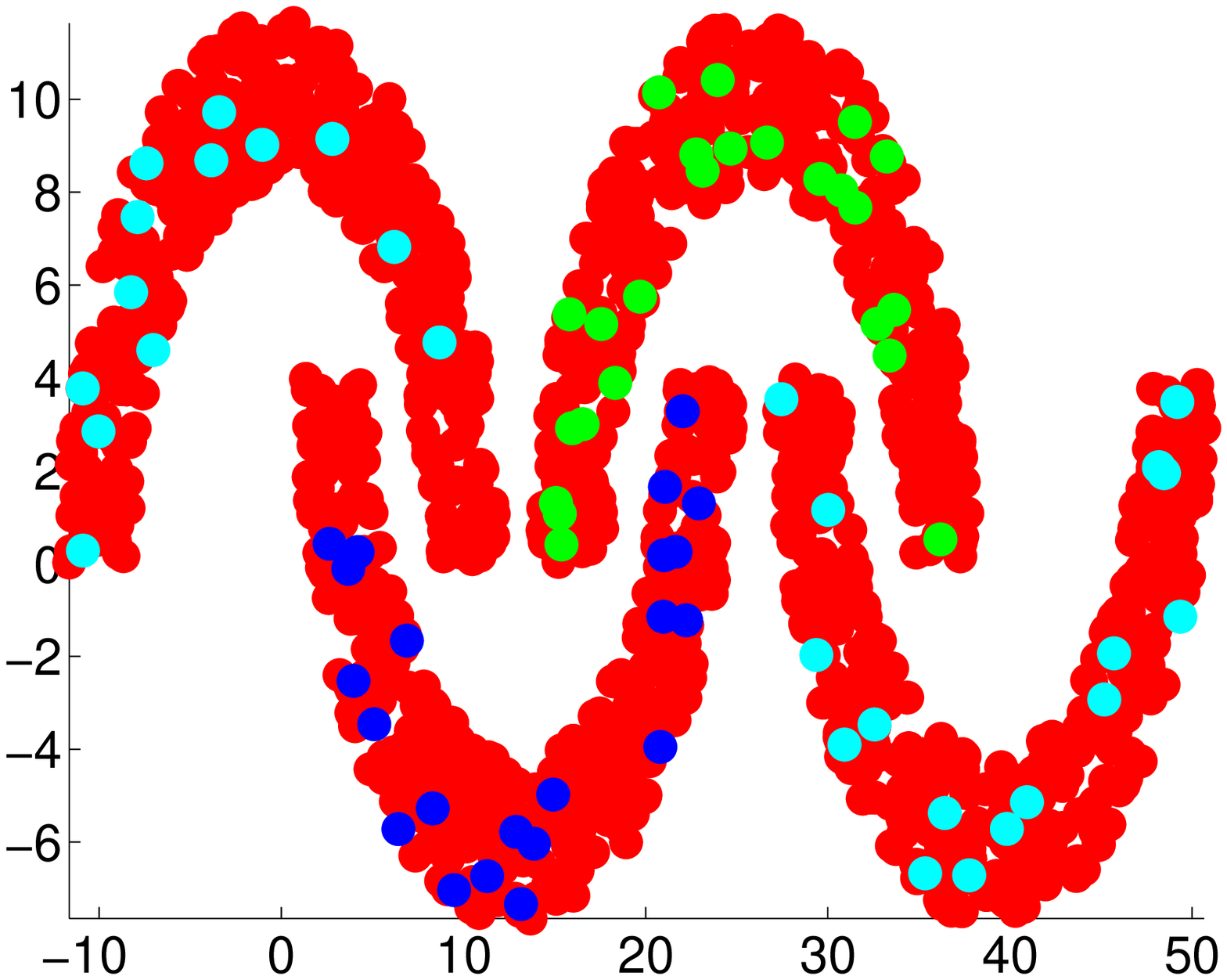}
\includegraphics[width=0.32\columnwidth]{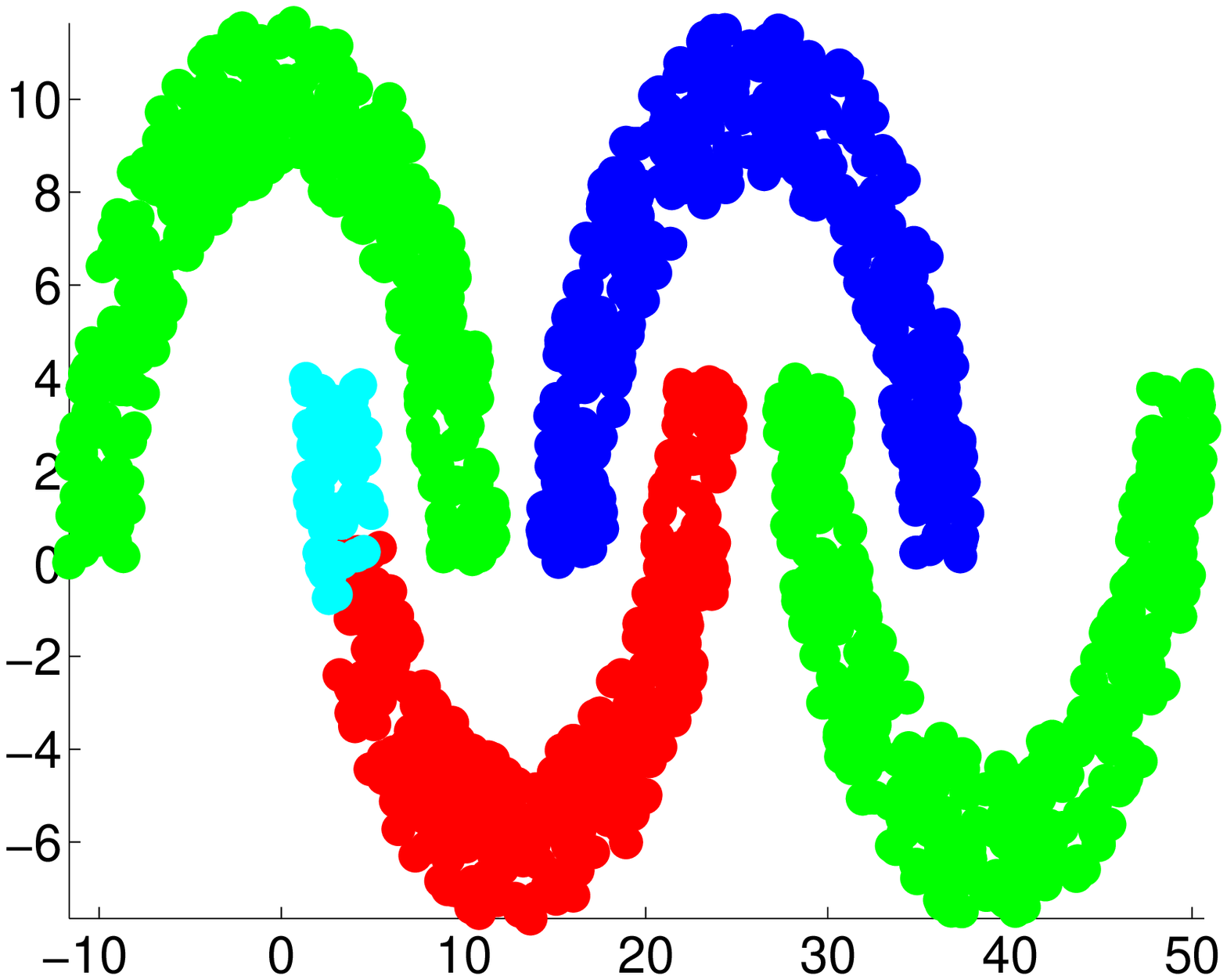}
\includegraphics[width=0.32\columnwidth]{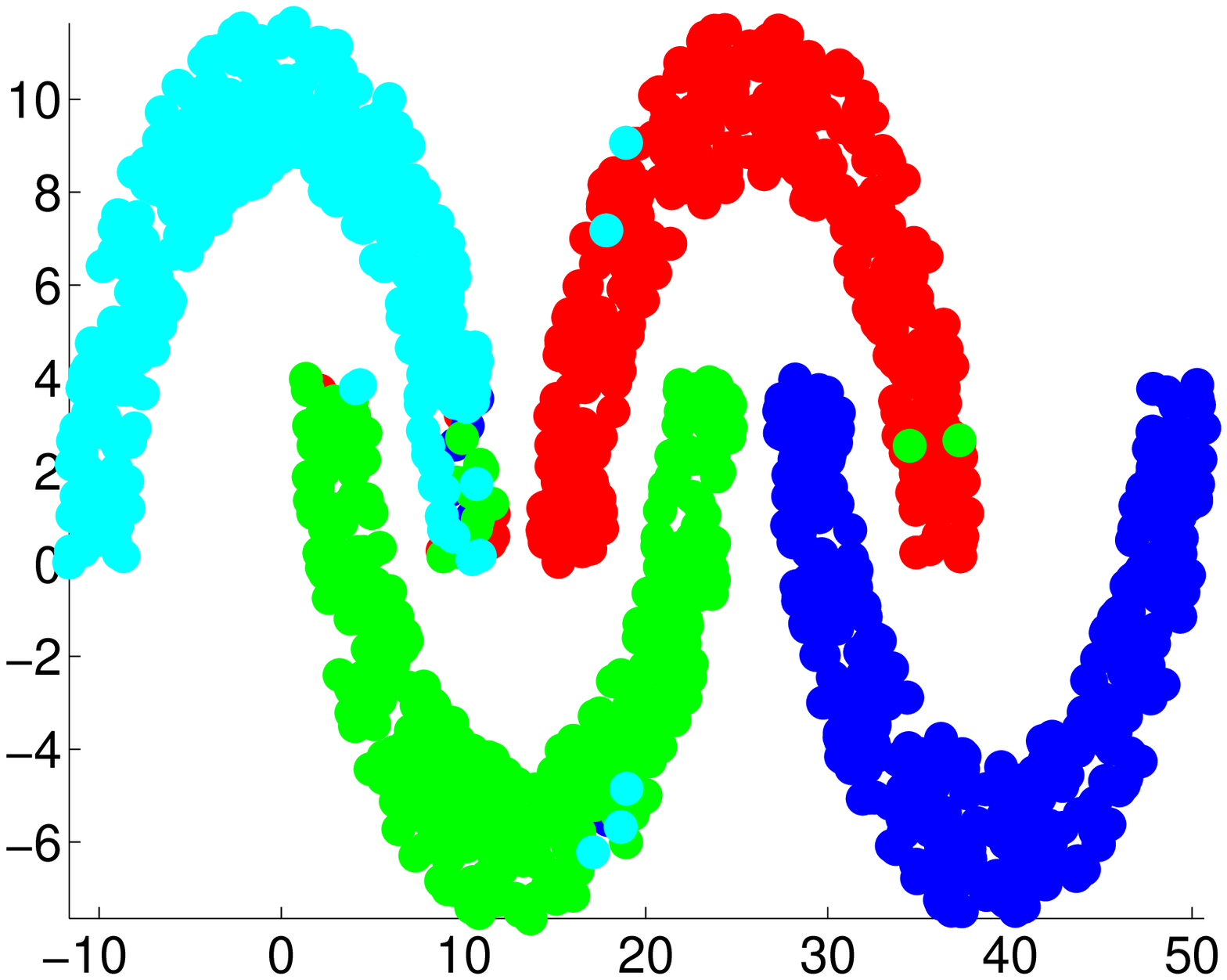}
\end{center}
\caption{Segmentation for a random instance of the Four-Moons data set with $75$ labels produced by \textbf{CSP} (left), \textbf{COSf} (middle) and \textbf{FAST-GE} (right).}
\label{fig:FourMoons1500_output}
\end{figure}

\begin{figure}[h]
\begin{center}
\includegraphics[width=0.48\columnwidth]{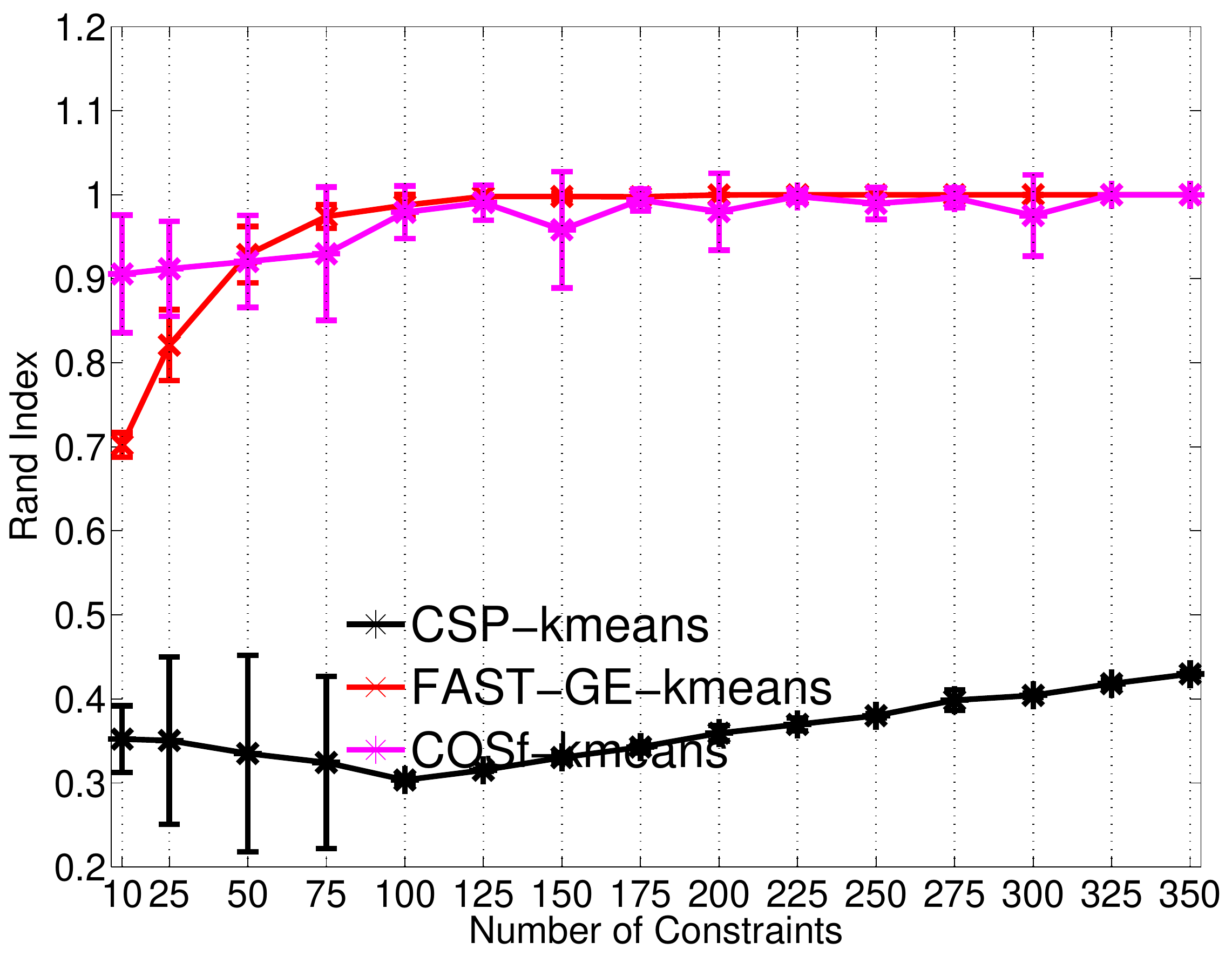}
\includegraphics[width=0.48\columnwidth]{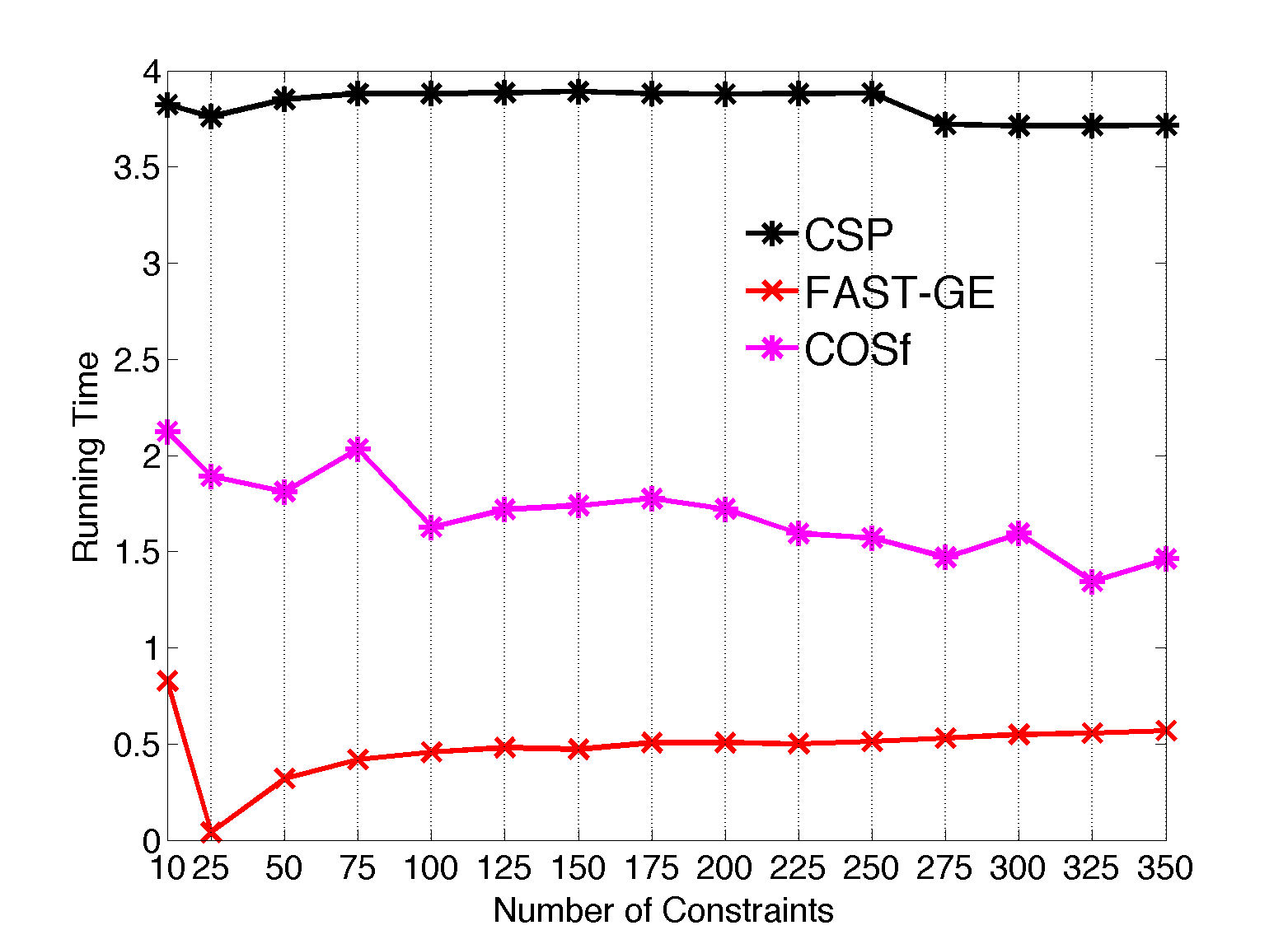}
\end{center}
\caption{Accuracy and running times for the Four-Moons data set, where the underlying graph given by the model NoisyKnn($n=1500, k=30,l=15$), for varying  number of constraints.
Time is in logarithmic scale. The bars indicate the variance in the output over random trials
using the same number of constraints.}
\label{fig:FourMoons1500_curves}
\end{figure}

\noindent \textbf{PACM.} Our second synthetic example is the somewhat more irregular \textit{PACM} graph, formed by a cloud of $n=426$ points in the shape of letters $\{ P,A,C,M\}$, whose topology renders the segmentation particularly challenging. The details about this data set are given in the section~\ref{sec:additional}. Here we only present a visualization of the obtained segmentations.

\begin{figure}[h]
\begin{center}
\includegraphics[width=0.30\columnwidth]{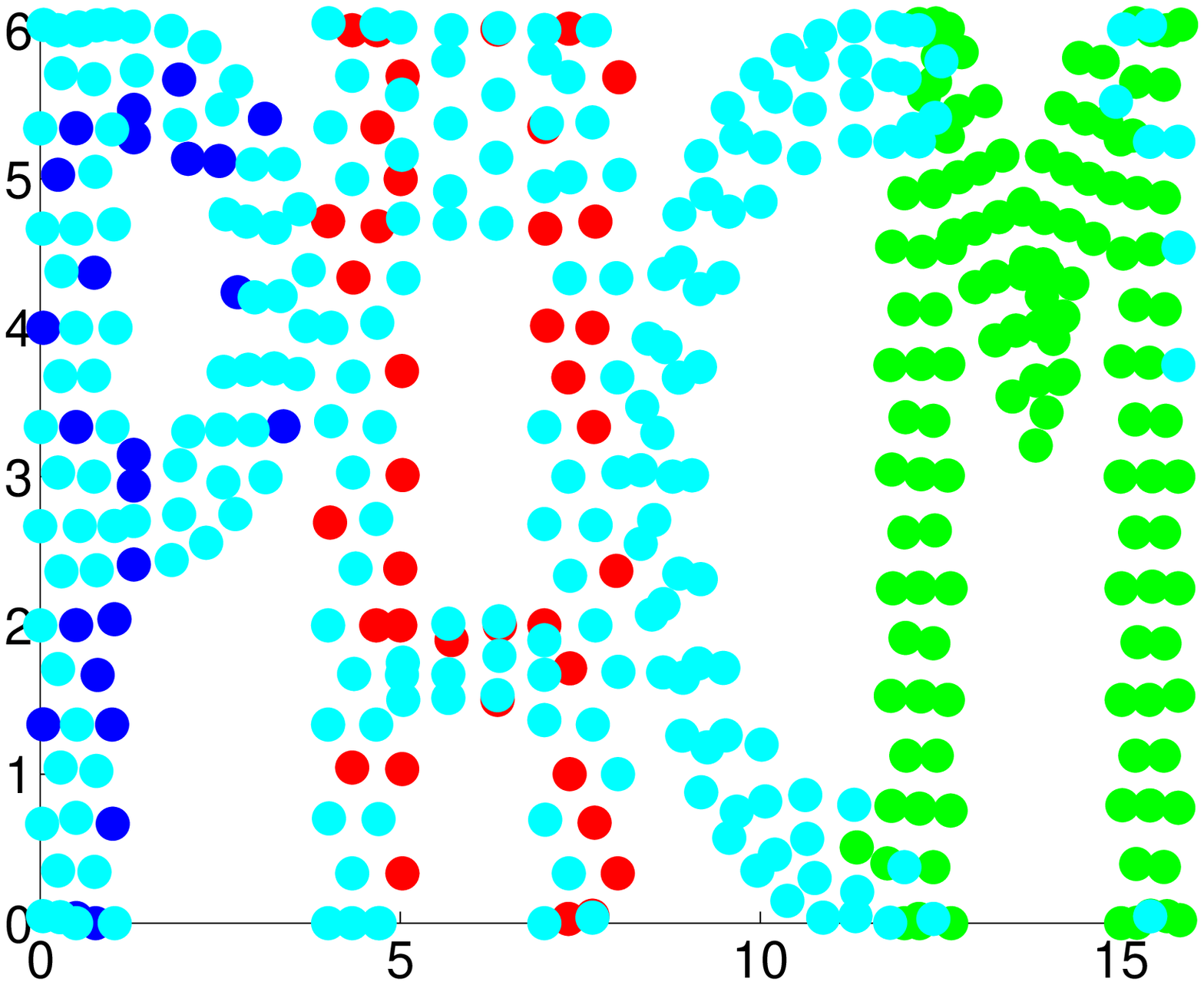}
\includegraphics[width=0.30\columnwidth]{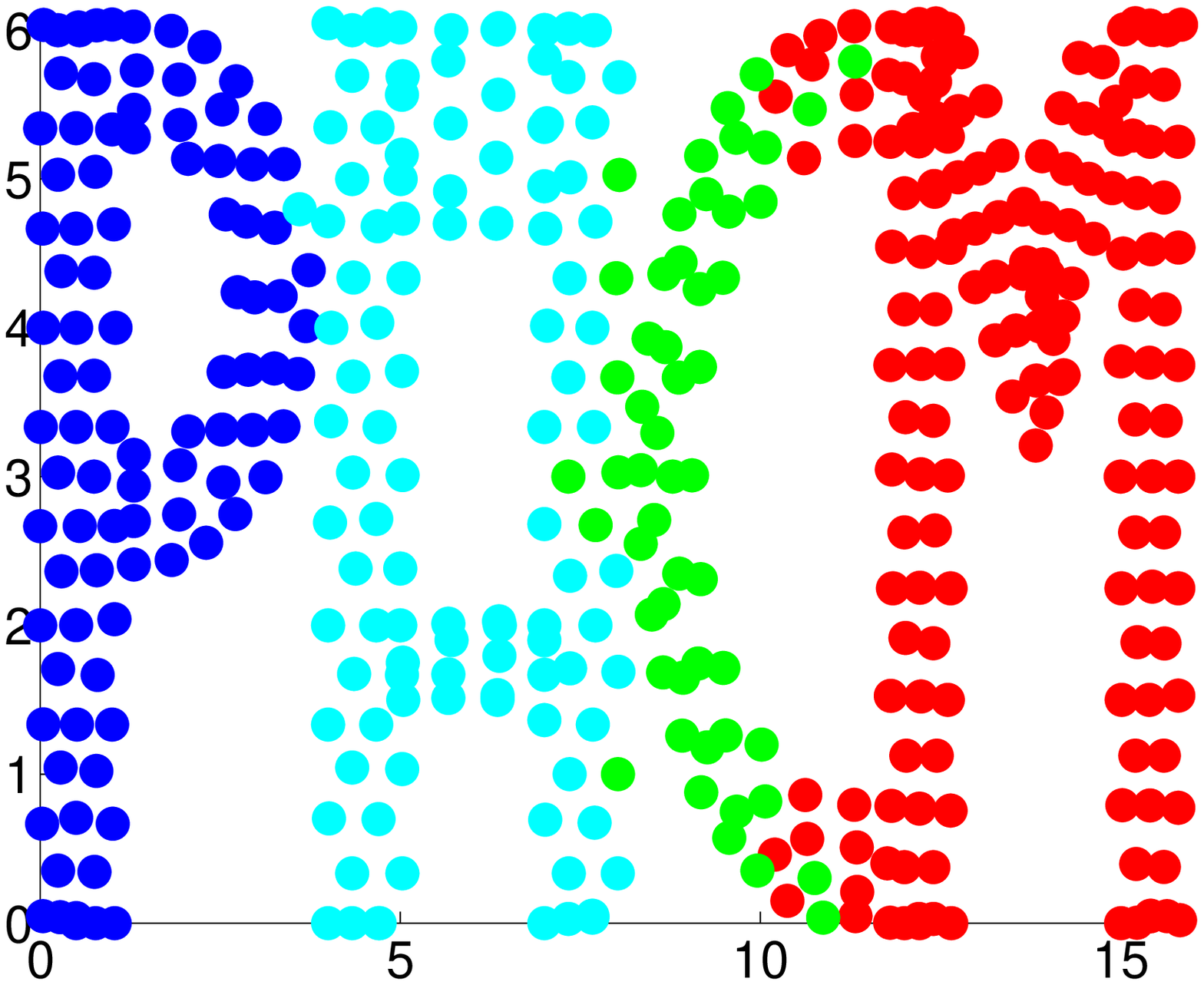}
\includegraphics[width=0.30\columnwidth]{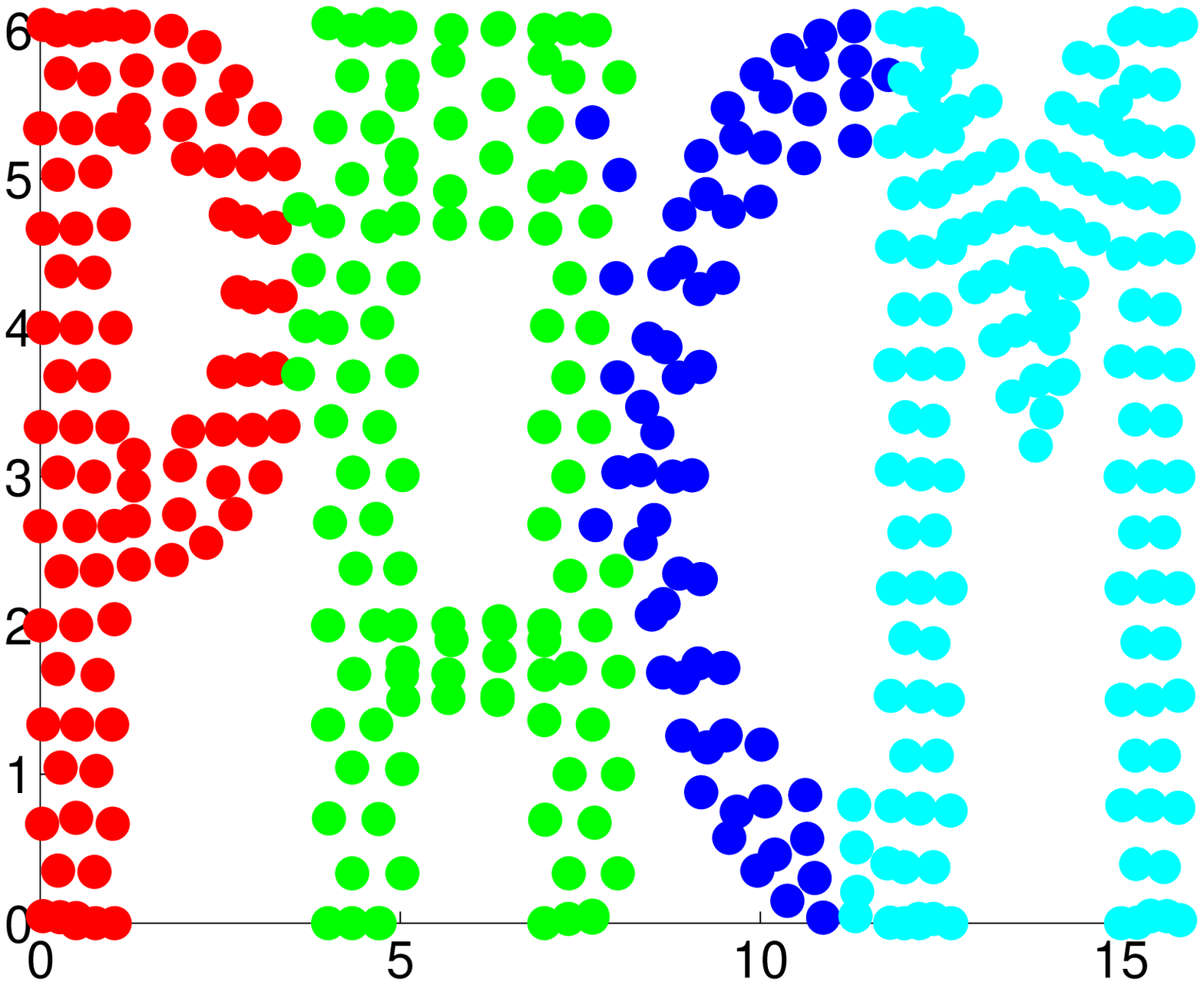}
\end{center}
\caption{Top: Segmentation for a random instance of the PACM data set with $125$ labels produced by \textbf{CSP} (left), \textbf{COSf} (middle) and \textbf{FAST-GE} (right)}
\label{fig:PACM_clusterings}
\end{figure}

\subsection{Image  Data}
In terms of real data, we consider two very different applications. Our first application is to segmentation of real images, where the  underlying grid graph is given by the affinity matrix of the image, computed using the RBF kernel based on the grayscale values.

We construct the constraints by assigning cluster-membership information to a very small number of the pixels, which are shown colored in the pictures below. The cluster-membership information is then turned into pairwise constraints in the obvious way. Our output is obtained by running $k$-means 20 times and selecting the best segmentation according to the $k$-means objective value.

\noindent \textbf{Patras.} Figure \ref{fig:PatrasLarge} shows the 5-way segmentation of an image with approximately 44K pixels, which our method is able to detect in under \textbf{3 seconds}. The size of this problem is prohibitive for \textbf{CSP}. The \textbf{COSf} algorithm runs in \textbf{40 seconds} and while it does better on the lower part of the image it erroneously merges two of the clusters (the red and the blue one) into a single region.

\begin{figure}[h]
\begin{center}
{\includegraphics[width=0.48\columnwidth]{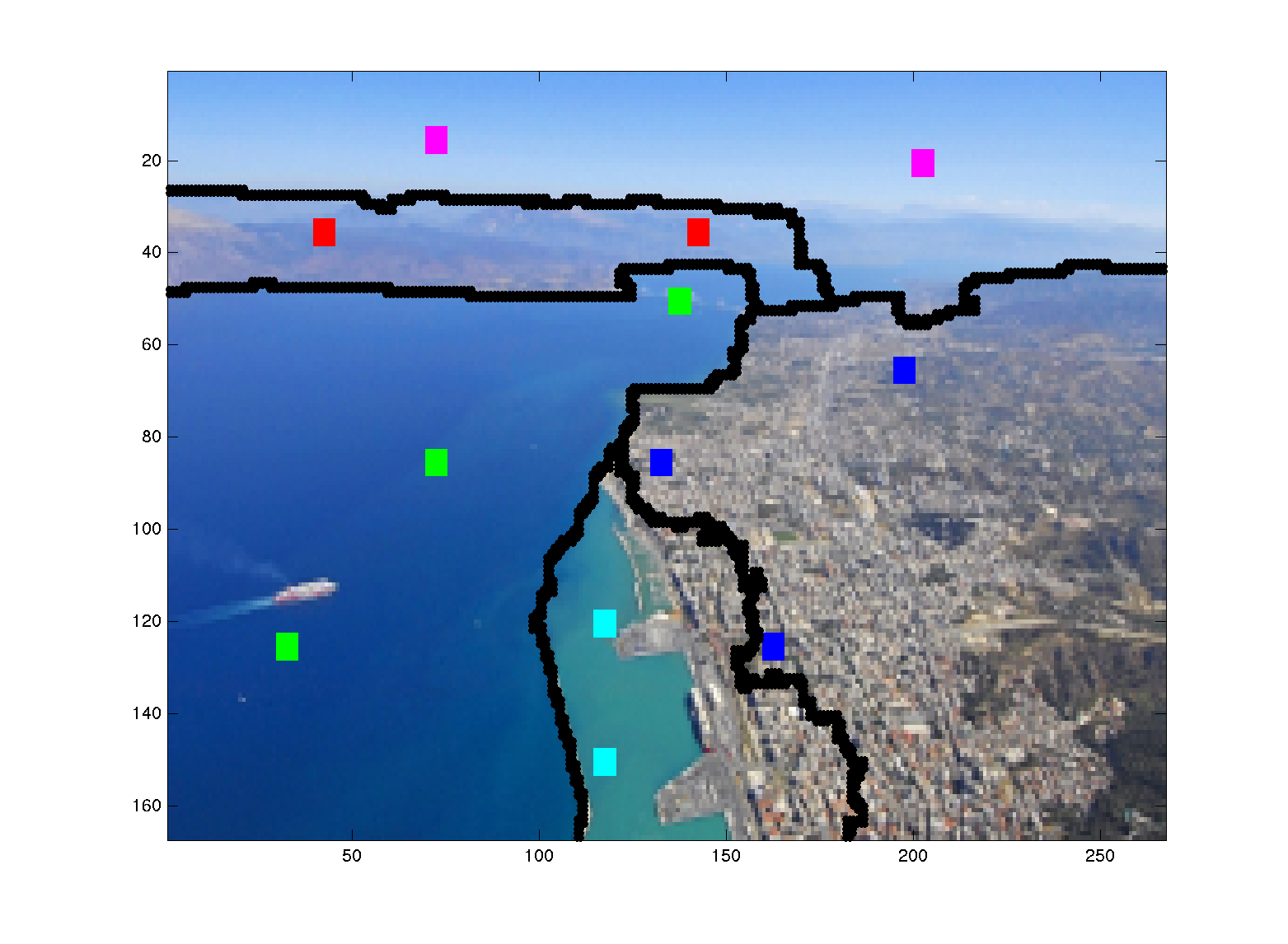} }
{\includegraphics[width=0.48\columnwidth]{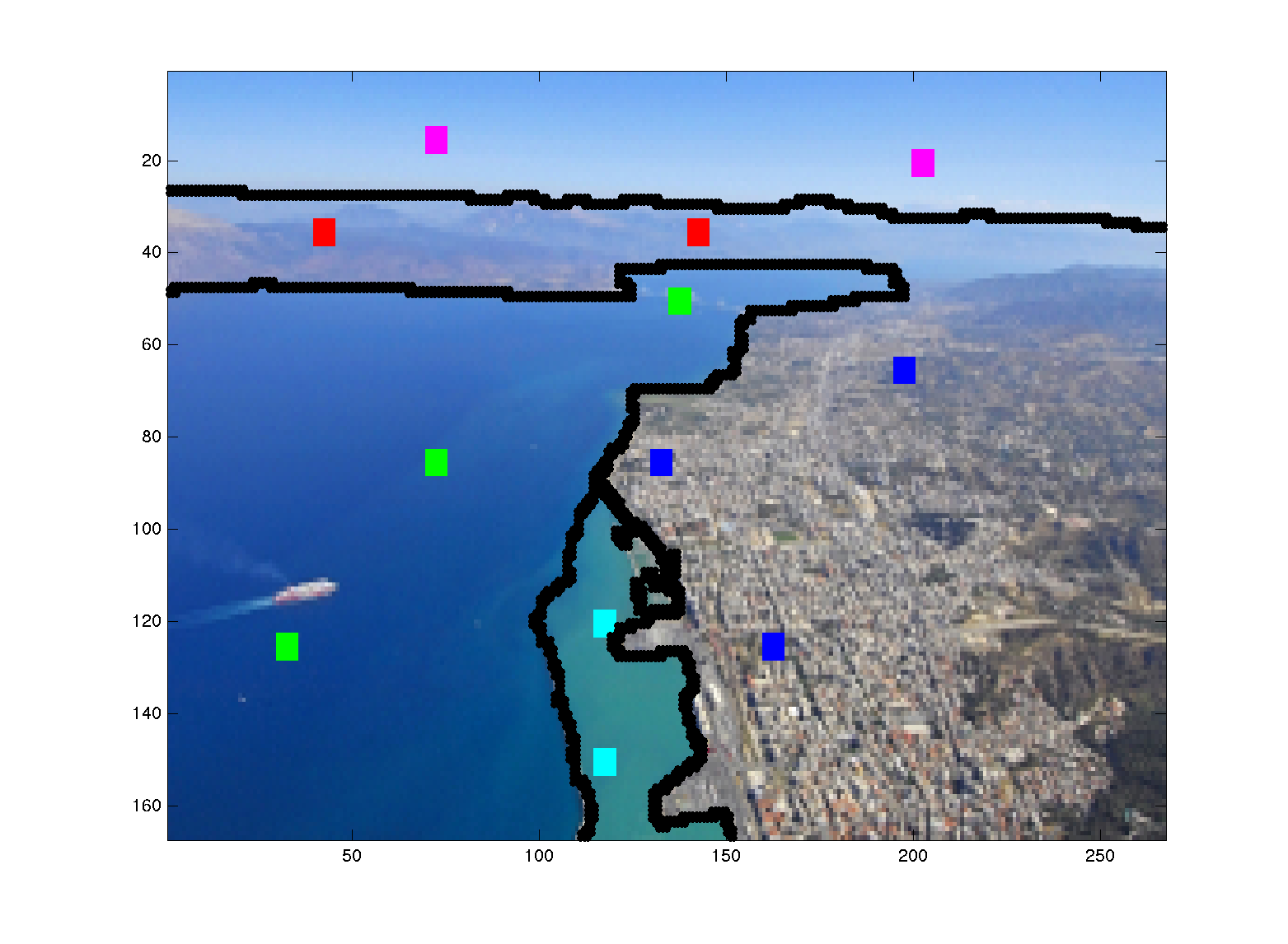} }
\includegraphics[width=0.48\columnwidth]{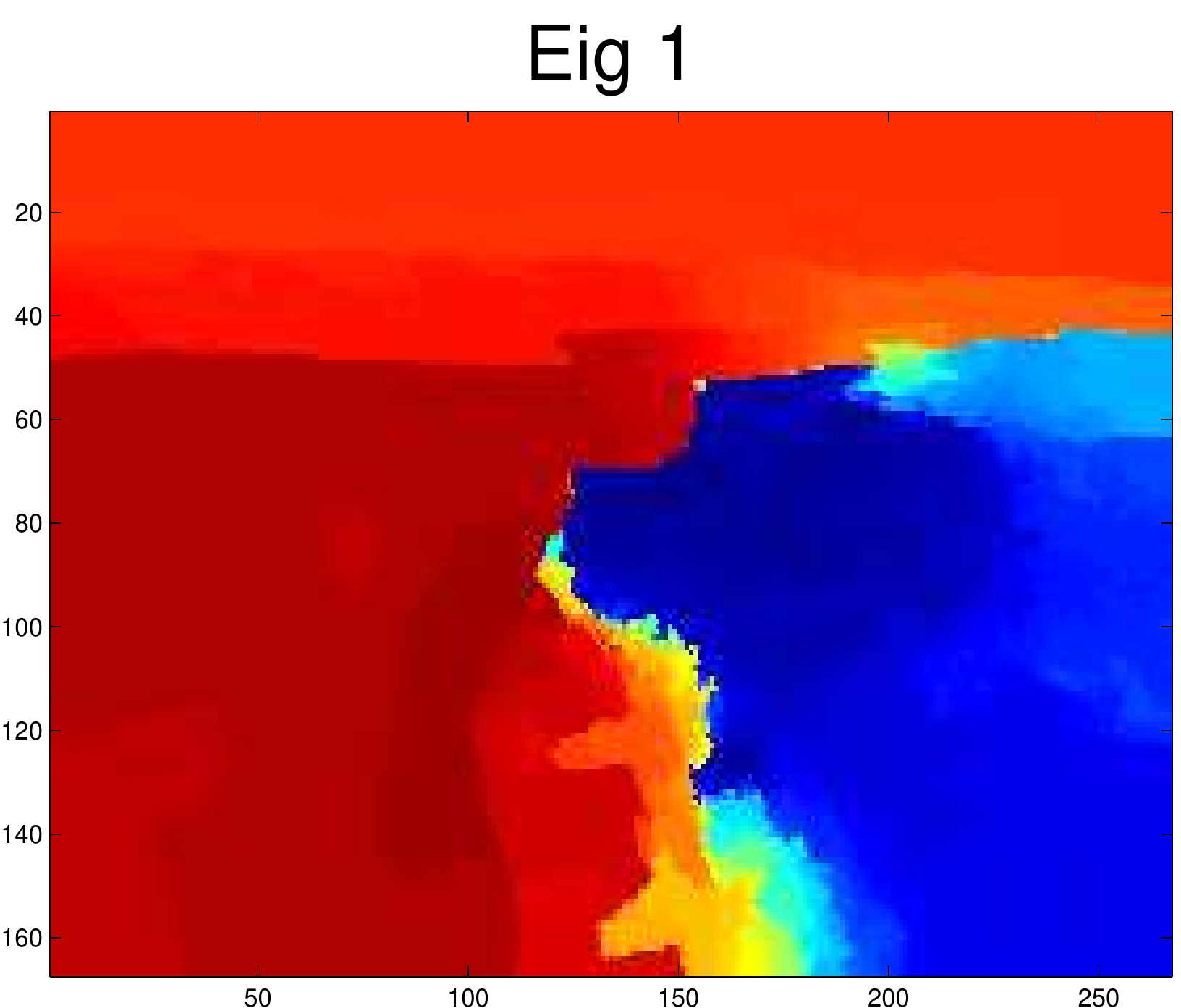}
\includegraphics[width=0.48\columnwidth]{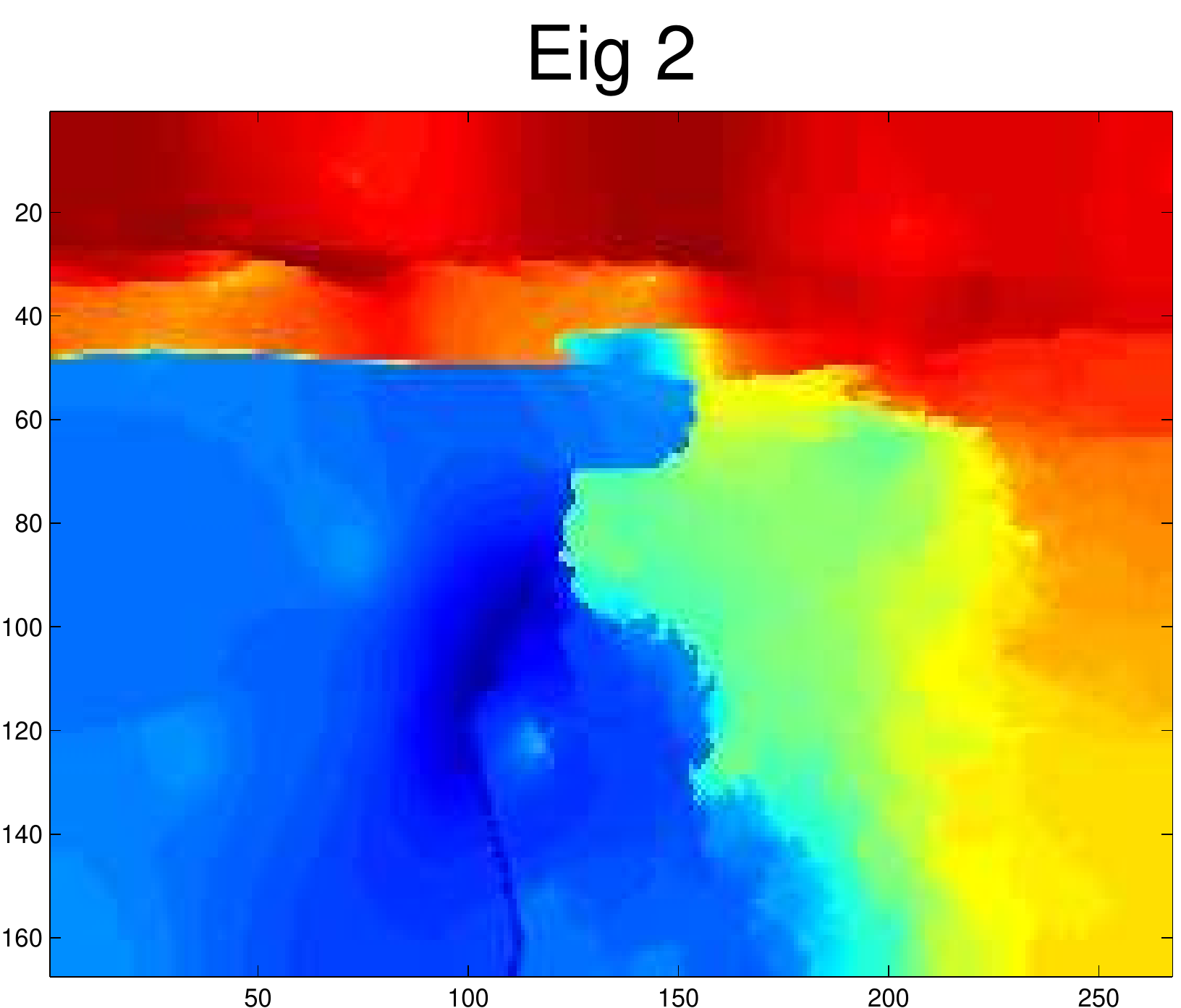}
\end{center}
\vspace{-3mm}
\caption{ \textit{Patras}: Top-left: Output of \textbf{FAST-GE}, in 2.8 seconds. Top-right: output of \textbf{COSf}, in 40.2 seconds. Bottom:  heatmaps for the first two eigenvectors computed by \textbf{FAST-GE}.
}
\label{fig:PatrasLarge}
\end{figure}

\noindent {\bf Santorini.}
In Figure \ref{fig:Santorini} we test our proposed method on the \textit{Santorini} image, with approximately $250K$ pixels. Our approach successfully recovers a 4-way partitioning, with few errors, in just \textbf{15 seconds}. Computing clusterings in data of this size is infeasible for \textbf{CSP}. The output of the \textbf{COSf} method, which runs in over \textbf{260 seconds}, is meaningless.

\begin{figure}[h]
\begin{center}
\includegraphics[width=0.48\columnwidth]{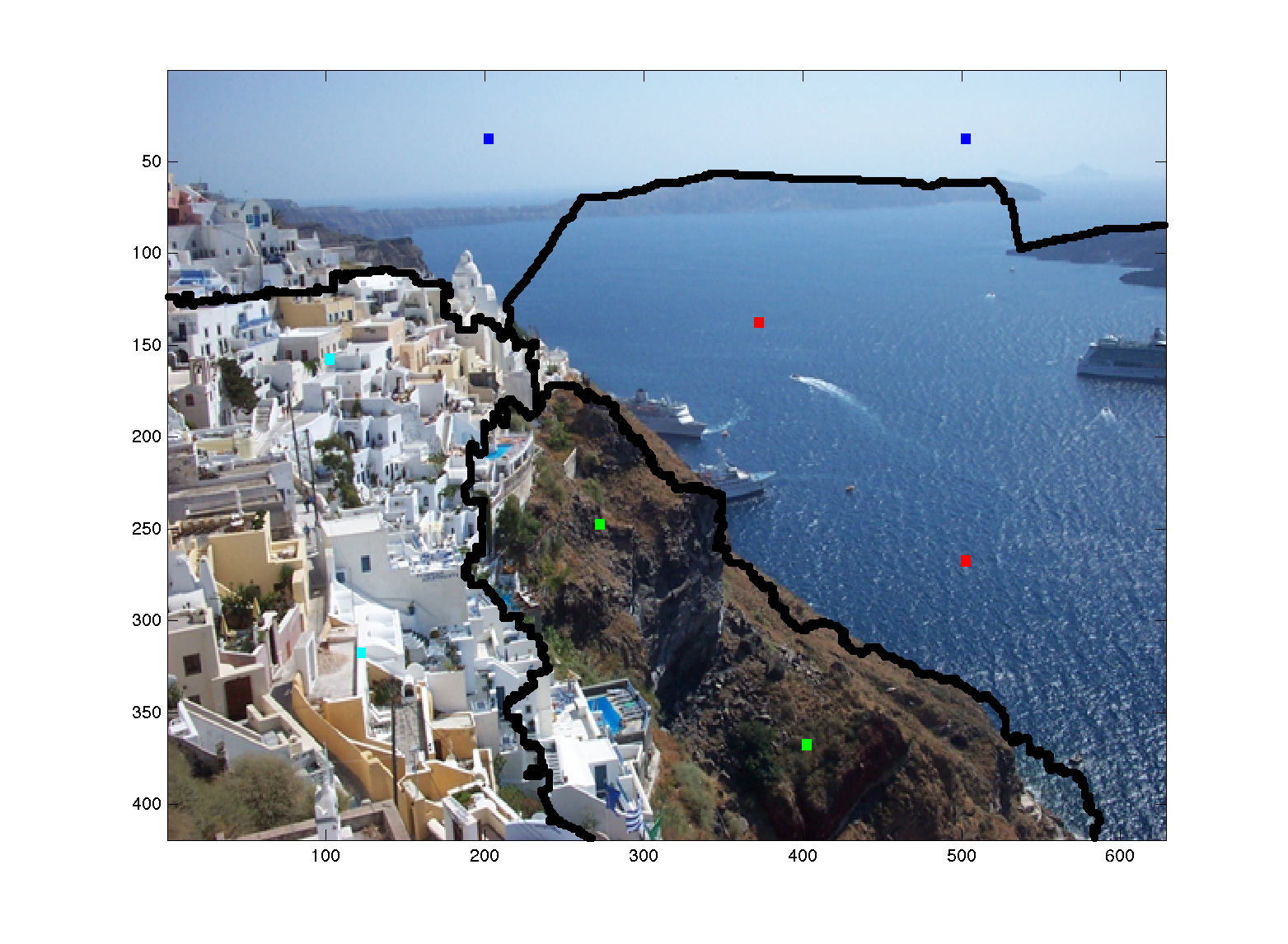}
\includegraphics[width=0.48\columnwidth]{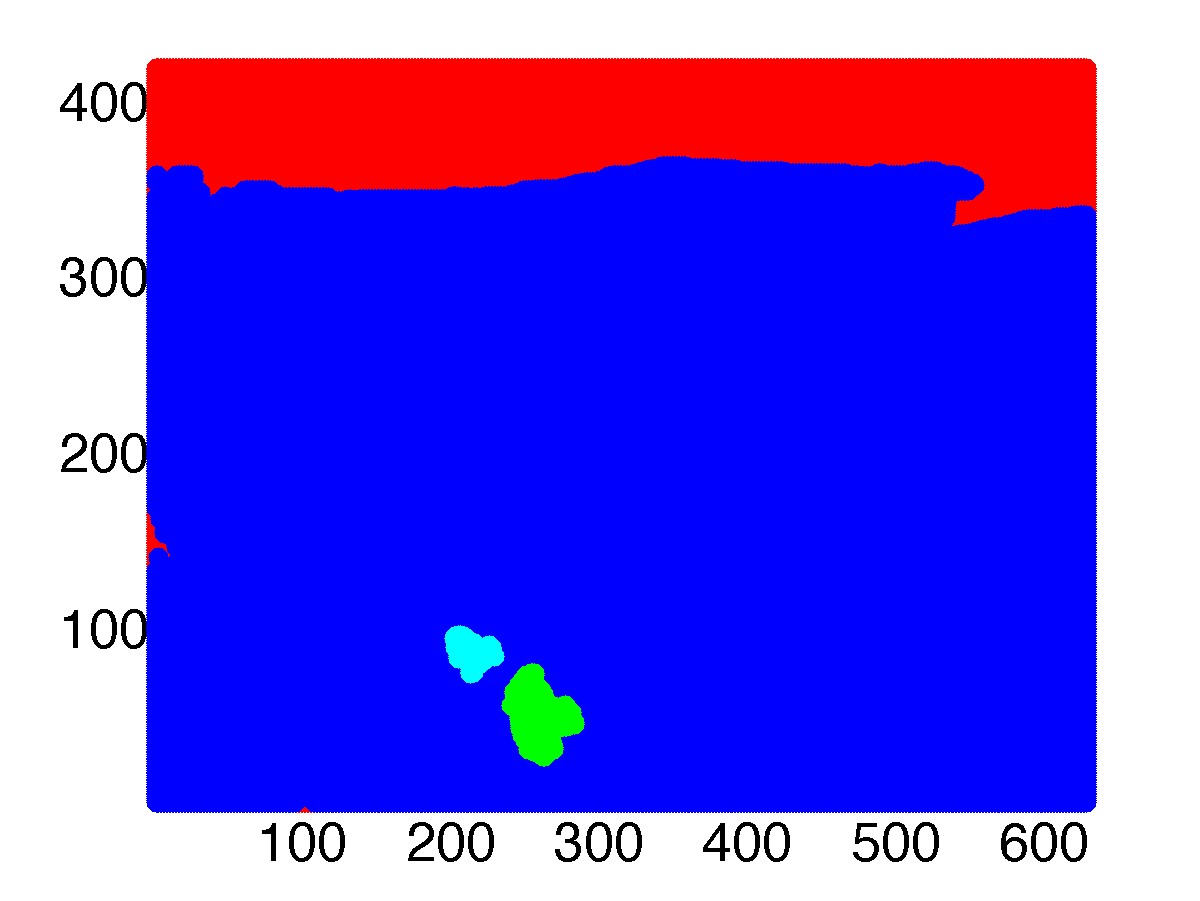}
\end{center}
\includegraphics[width=0.48\columnwidth]{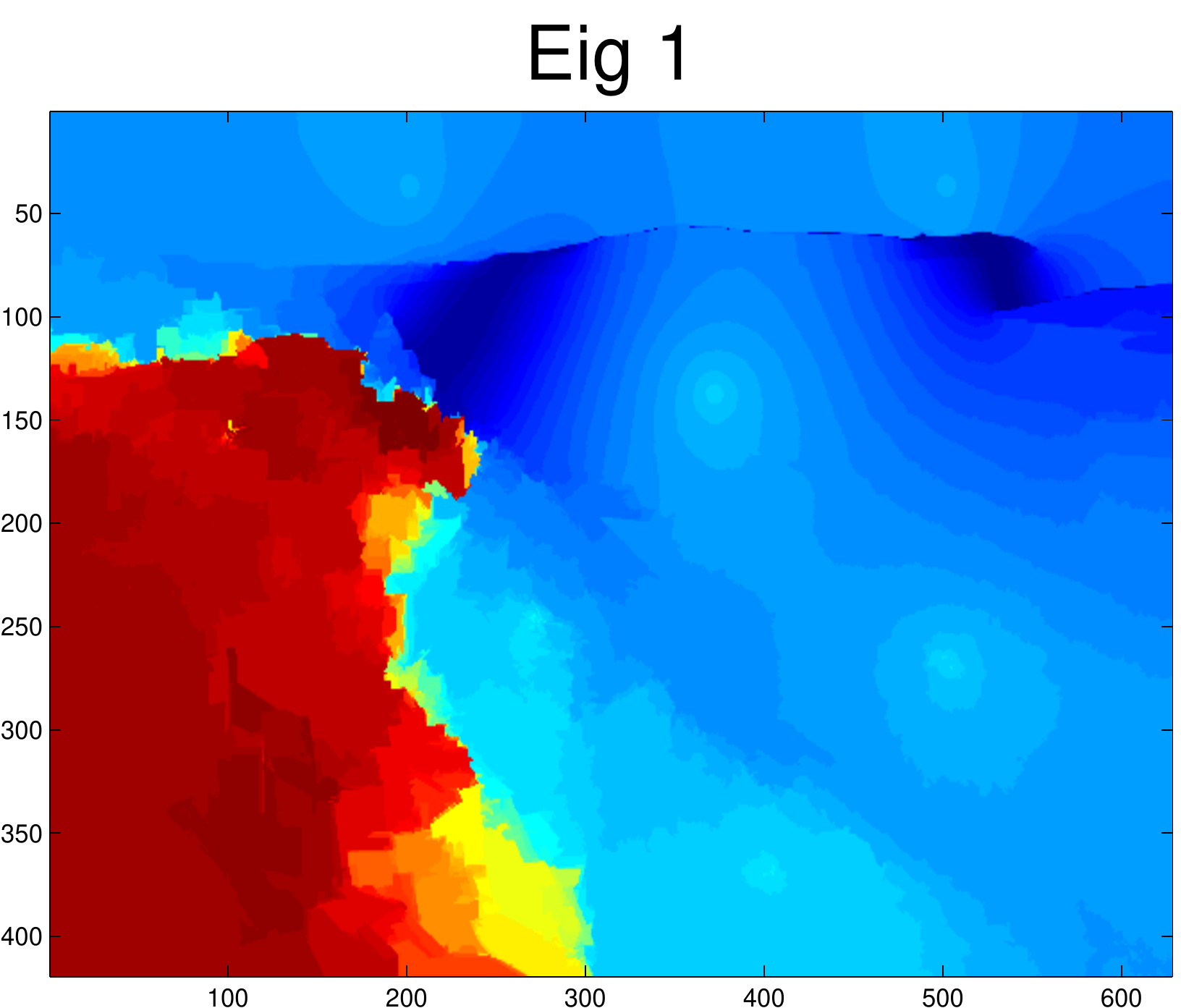}
\includegraphics[width=0.48\columnwidth]{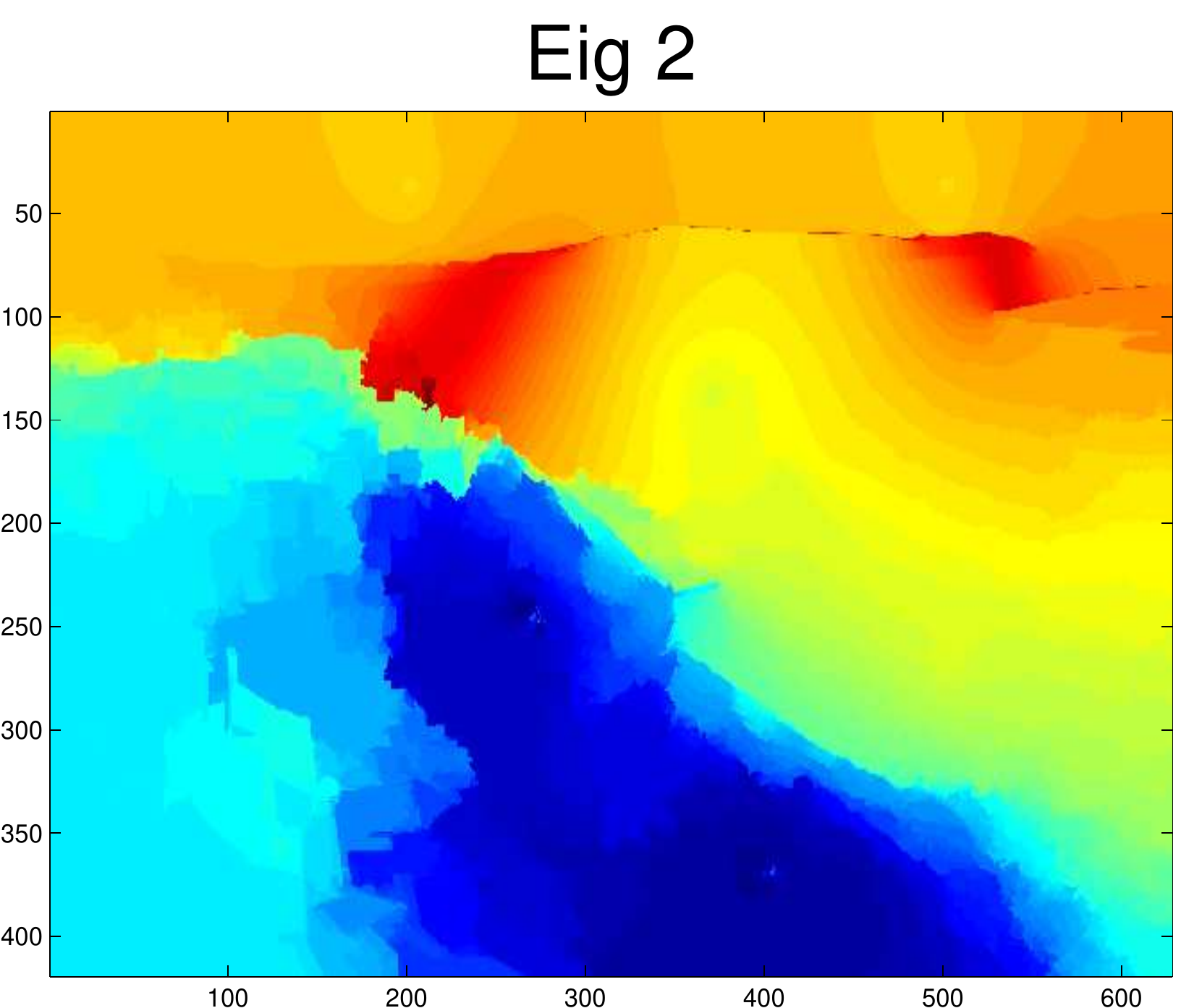}
\vspace{-3mm}
\caption{ \textit{Santorini}: Left: output of \textbf{FAST-GE}, in 15.2 seconds. Right: output of \textbf{COSf}, in 263.6  seconds. Bottom: heatmaps for the first two eigenvectors computed by  \textbf{FAST-GE}.}
\label{fig:Santorini}
\end{figure}

\noindent \textbf{Soccer.} In Figure \ref{fig:Soccer} we consider one last \textit{Soccer} image, with approximately 1.1 million pixels. We compute a 5-way partitioning using the \textbf{Fast-GE}  method in just \textbf{94 seconds}. Note that while k-means clustering hinders some of the details in the image, the individual eigenvectors are able to capture finer details, such as the soccer ball for example, as shown in the two bottom plots of the same Figure \ref{fig:Soccer}. The output of the \textbf{COSf} method is obtained in {\bf 25 minutes} and is again  meaningless.

\begin{figure}[h]
\begin{center}
\includegraphics[width=0.48\columnwidth]{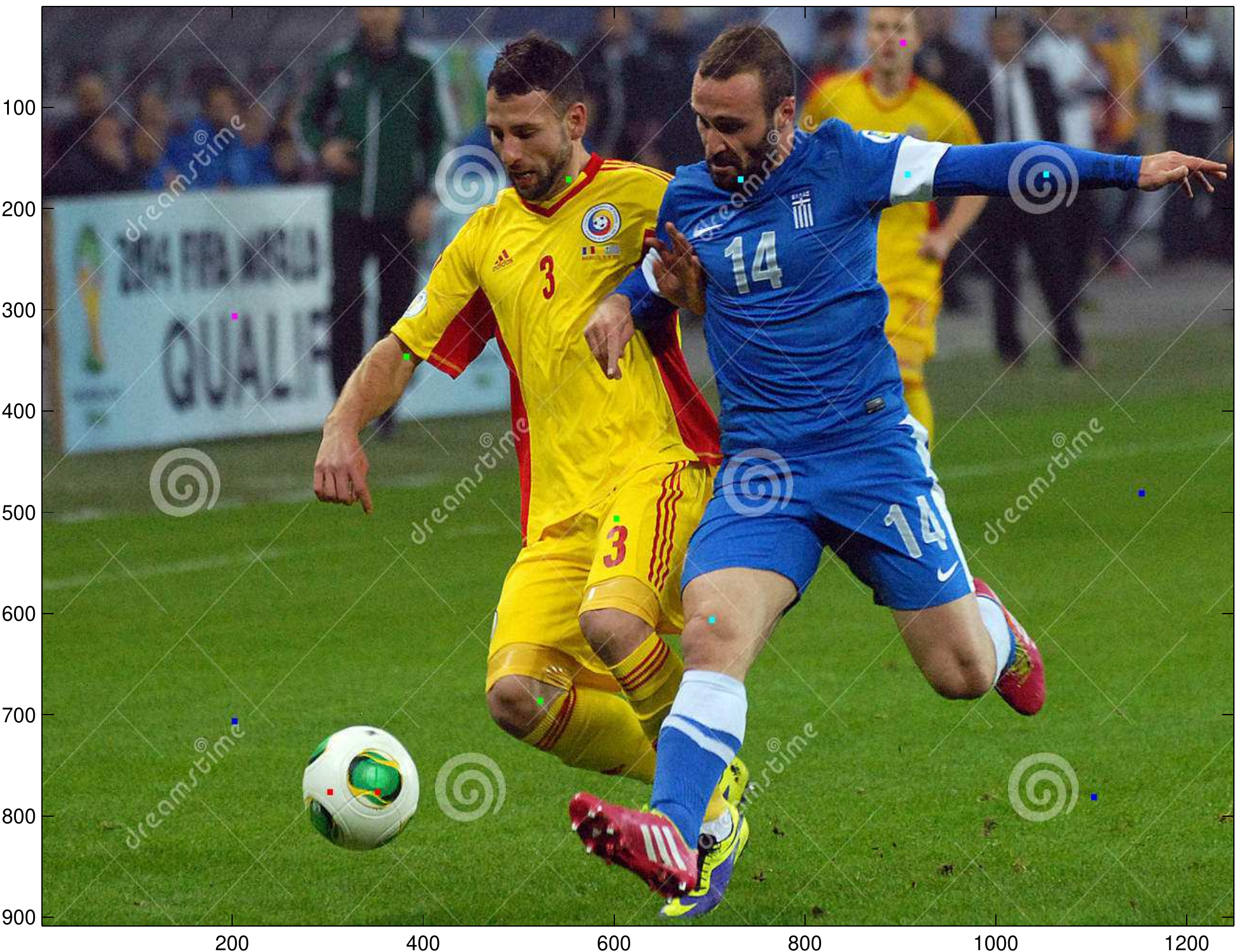}
\includegraphics[width=0.48\columnwidth]{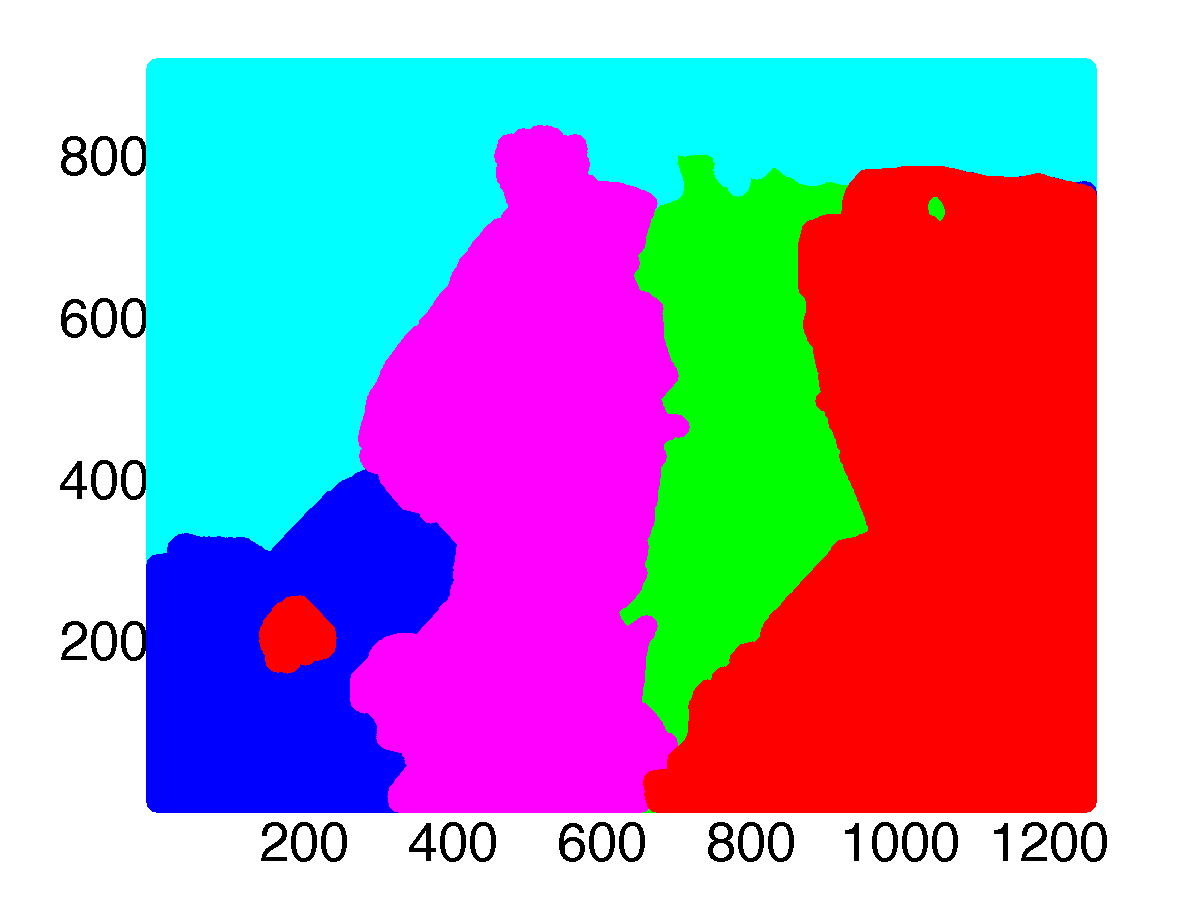}
\includegraphics[width=0.48\columnwidth]{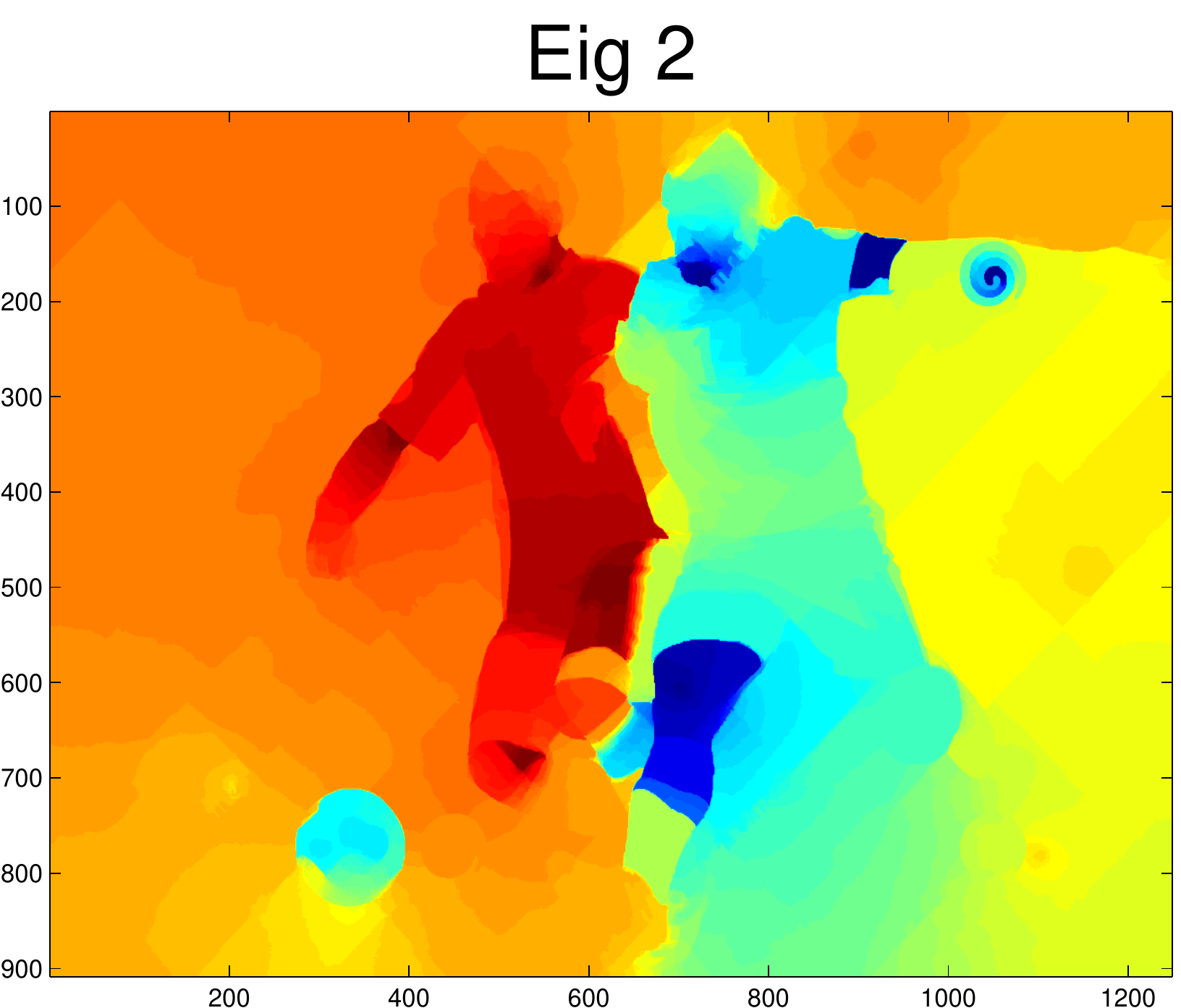}
\includegraphics[width=0.48\columnwidth]{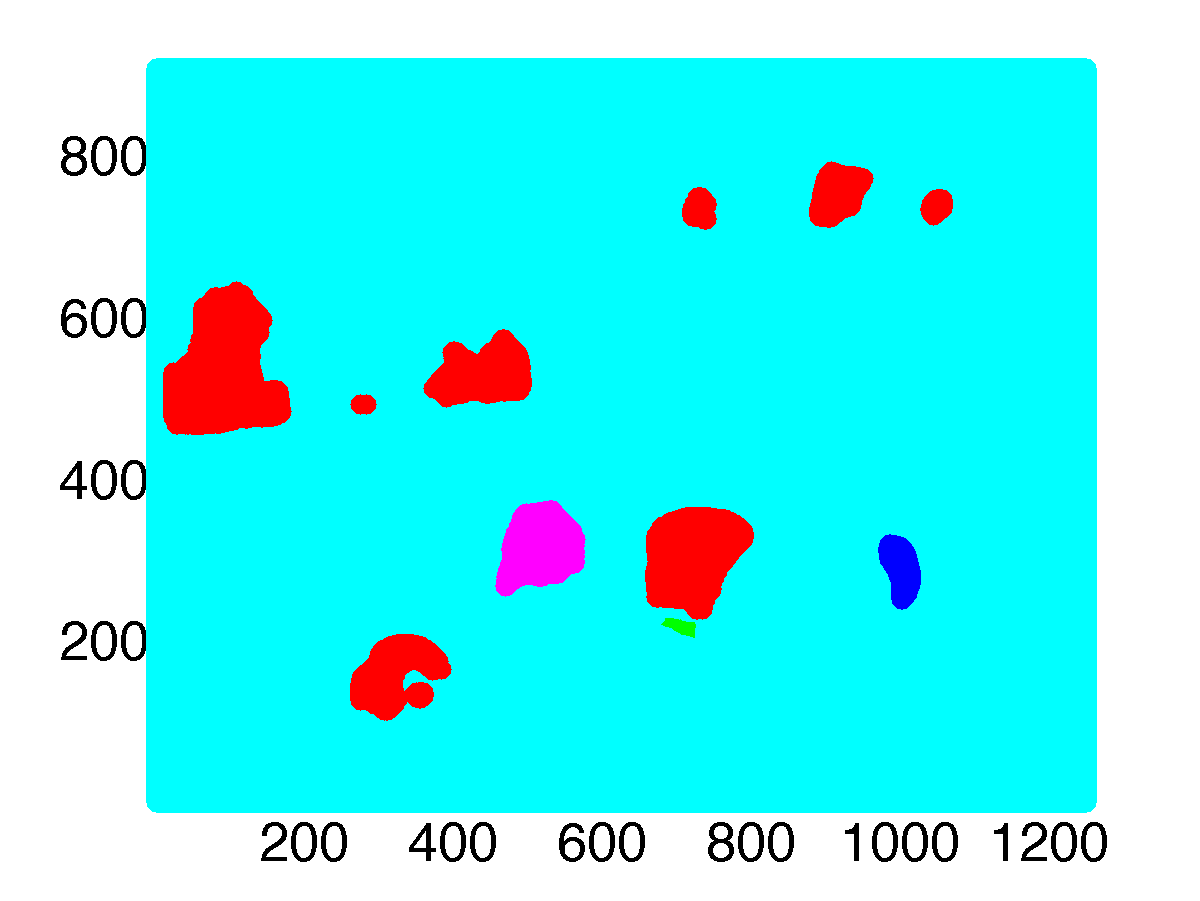}
\end{center}
\vspace{-3mm}
\caption{Top-right: output of \textbf{FAST-GE}, in under 94 seconds. Bottom-right: output of \textbf{COSf} in 25 minutes. Bottom-left: heat-maps of eigenvectors. }
\label{fig:Soccer}
\end{figure}

\subsection{Friendship Networks}
Our final data sets represent Facebook networks in American colleges. The work in~\cite{traud2012Facebook} studies the structure of Facebook  networks at one hundred American colleges and universities at a single point in time (2005) and investigate the community structure at each institution, as well as the impact and correlation of various self-identified user characteristics (such as residence, class year, major, and high school) with the identified network communities. While at many institutions, the community structures are organized almost exclusively according to class year, as pointed out in \cite{TraudSIAMreview},
other institutions are known to be organized almost exclusively according to its undergraduate \textit{House} system (dormitory residence), which is very well reflected in the identified communities. It is thus a natural assumption to consider the dormitory affiliation as the ground truth clustering, and aim to recover this underlying structure from the available friendship graph and any available constraints. We add constraints to the clustering problem by sampling uniformly at random nodes in the graph, and the resulting pairwise constraints are generated depending on whether the two nodes belong to the same cluster or not. In order for us to be able to compare to the computationally expensive \textbf{CSP} method, we consider two small-sized schools,
Simmons College ($  n=850$, $\bar{d}=36$, $k=10$) and Haverford College ($n=1025$, $\bar{d}=72$, $k=15$), where $\bar{d}$ denotes the average degree in the graph and $k$ the number of clusters. For both examples, \textbf{FAST-GE} yields more accurate results than both \textbf{CSP} and \textbf{COSf}, and does so at a much smaller computational cost.

\begin{figure}[h]
\begin{center}
\includegraphics[width=0.42\columnwidth]{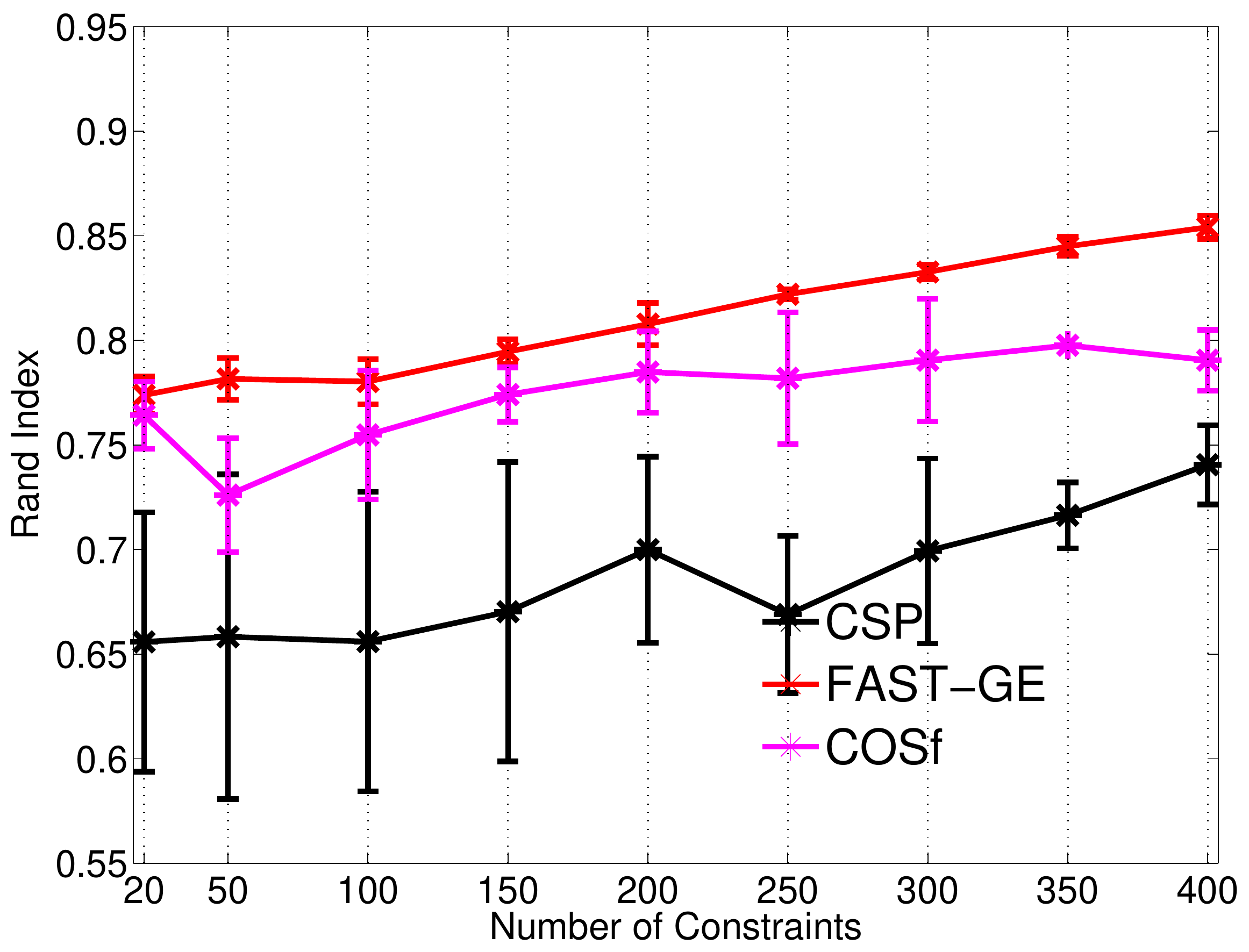}
\includegraphics[width=0.42\columnwidth]{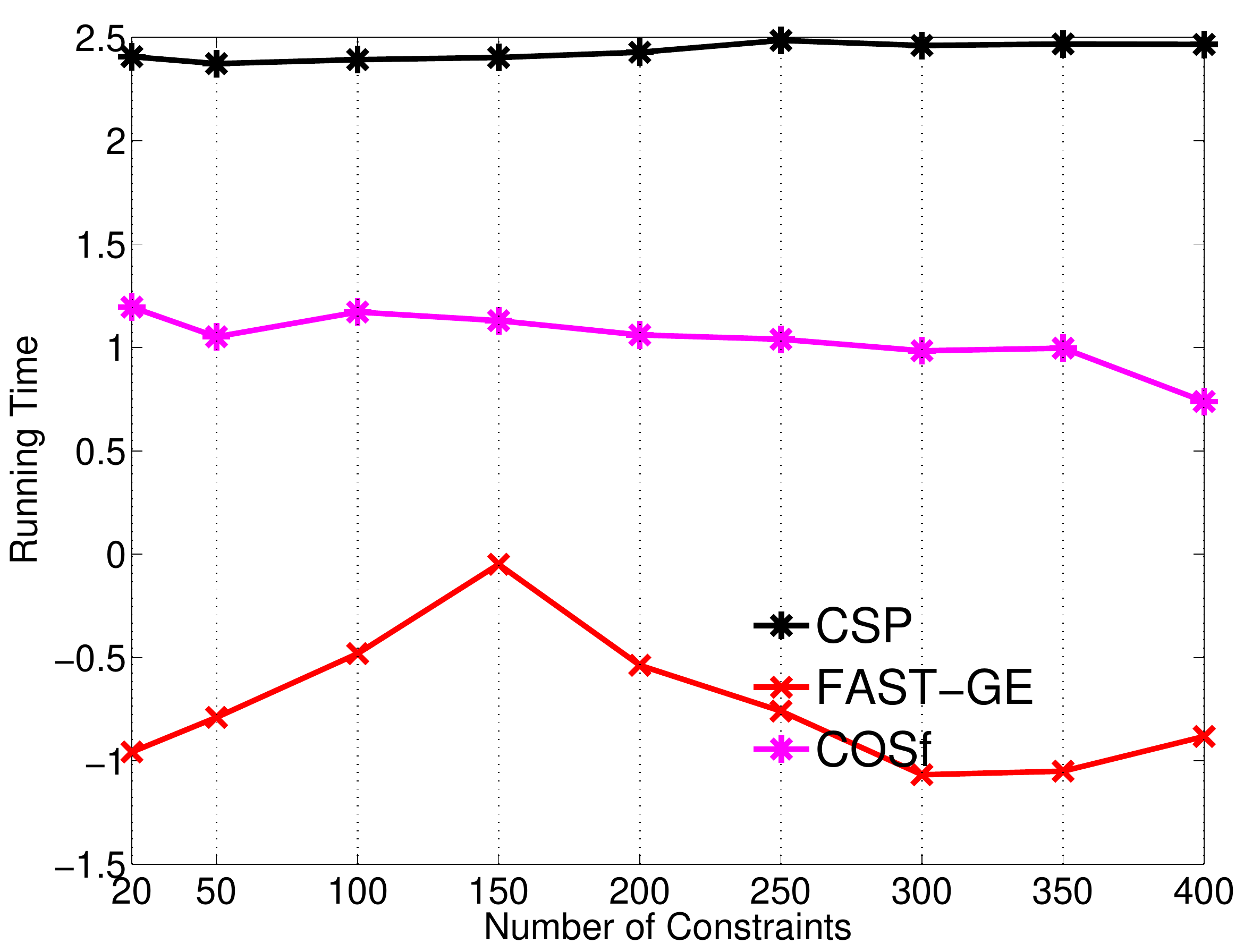}
\includegraphics[width=0.42\columnwidth]{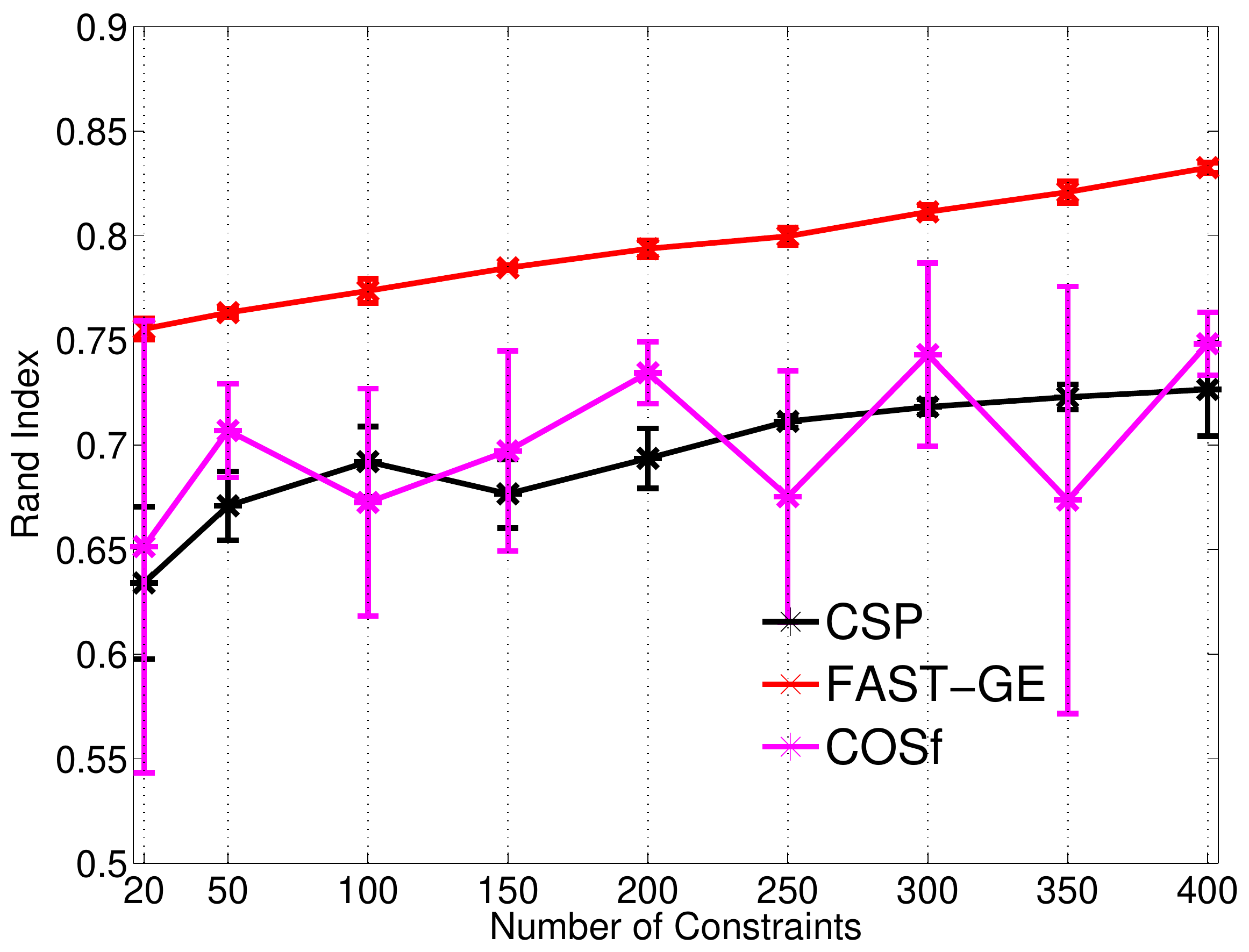}
\includegraphics[width=0.42\columnwidth]{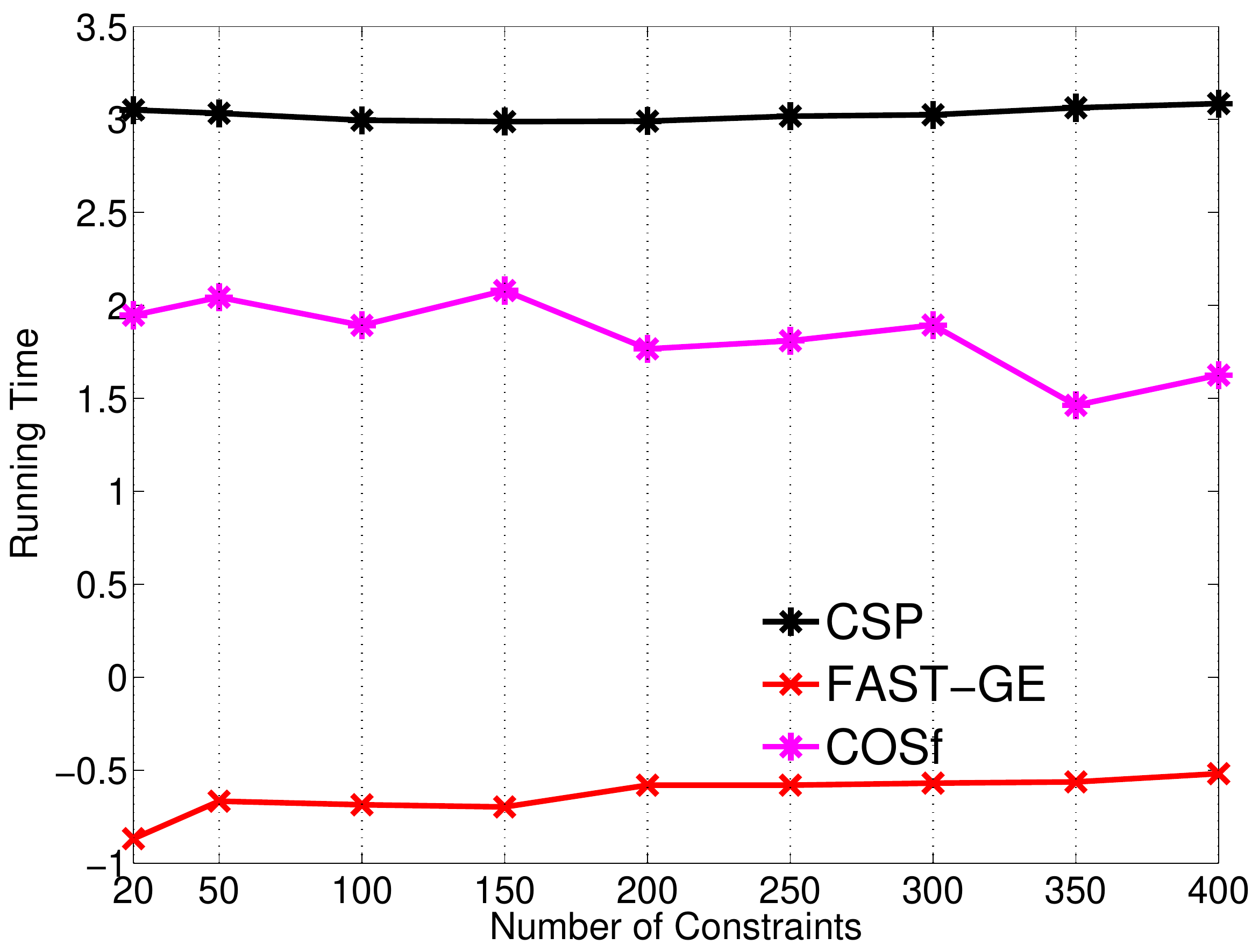}
\end{center}
\vspace{-1mm}
\caption{\textit{Facebook networks}. Top: accuracy and running times for the Simmons College ($  n=850$, $\bar{d}=36$, $k=10$). Bottom: accuracy and running times for Haverford College ($n=1025$, $\bar{d}=72$, $k=15$). Time is in logarithmic scale. }
\label{fig:Facebook}
\end{figure}

%% file: conclusions.tex
\section{Final Remarks}

We presented a spectral method that reduces constrained
clustering into a generalized eigenvalue problem in which both matrices
are Laplacians. This  offers
two advantages that are not simultaneously shared by
any of the previous methods: an efficient implementation
and an approximation guarantee for the 2-way partitioning problem
in the form of a generalized Cheeger inequality. In practice
this translates to a method that is at least 10x faster than some of
the best existing algorithms, while producing output of superior quality.
Its speed makes our method a good candidate for some type of
iteration, e.g. as in~\cite{TolliverM06},
or interactive user feedback, that would further improve its output.

We view the Cheeger inequality we
presented in section~\ref{sec:cheeger}
as indicative of the rich
mathematical properties of generalized Laplacian
eigenvalue problems. We expect that tighter
versions
are to be discovered, along the lines of~\cite{KLLGT13}.
Finding $k$-way generalizations of the Cheeger inequality,
as in~\cite{Lee12}, poses an interesting open problem.

%% file: supplementary-file.tex
\section{Proof of the Generalized Cheeger Inequality} \label{sec:proof}

We begin with two Lemmas.

\begin{lemma} \label{th:splitting}
For all $a_i, b_i > 0$ we have
    $$\frac{\sum_i {a_i}}{\sum_i b_i} \geq \min_i \left\{ \frac{a_i}{b_i}\right\}.$$
\end{lemma}

\begin{lemma} \label{th:DvsDG}
Let $G$ be a graph, $d$ be the vector containing the degrees of the vertices,
and $D$ be corresponding diagonal matrix. For
all vectors $x$ where $x^Td = 0$ we have
$$
    x^T D x = x^T L_{D_G} x,
$$
where $D_G$ is the demand graph for $G$.
\end{lemma}
\begin{proof}
Let $d$ be the vector consisting of the entries
along the diagonal of $D$. By definition, we have
$$
   L_{D_G} = D - \frac{dd^T}{vol(V)}.
$$
The lemma follows.
\end{proof}

We prove the following theorem.
\begin{theorem}
\label{thm:generalizedcheeger}

Let $G$ and $H$ be any two weighted graphs and $D$ be
the vector containing the degrees of the vertices in $G$.
F any vector $x$ such that $x^Td =0$, we have

\[
\frac{x^TL_Gx}{x^TL_Hx} \geq \phi(G,D_G) \cdot \phi(G, H)/4,
\]
where $D_G$ is the demand graph of $G$.
A cut meeting the guarantee of the inequality can
be obtained via a Cheeger sweep on $x$.
\end{theorem}

Let $V^{-}$ denote the set of $u$ such that $x_u \leq 0$ and
$V^{+}$ denote the set such that $x_u > 0$.
Then we can divide $E_G$ into two sets: $E_G^{same}$ consisting
of edges with both endpoints in $V^{-}$ or $V^{+}$, and
$E_G^{dif}$ consisting of edges with one endpoint in each.
In other words:
\begin{align*}
& E_G^{dif} = \delta_G\left(V^{-}, V^{+}\right),\text{ and}\\
& E_G^{same} = E_G \setminus E_G^{dif}.
\end{align*}

We also define $E_H^{dif}$ and $E_H^{same}$ similarly.

We first show a lemma which is identical to one used in the proof
of Cheeger's inequality~\cite{chung1}:

\begin{lemma} \label{lem:abssqr}
Let $G$ and $H$ be any two weighted graphs on the same vertex set $V$
partitioned into $V^{-}$ and $V^{+}$. For any vector $x$ we have
\[
\frac{ \sum_{uv \in E_G^{same}}w_G\left(u,v\right)\left|x_u^2-x_v^2\right| +
\sum_{uv \in E_G^{dif}}w_G(u, v) \left(x_u^2 + x_v^2\right) }{x^T L_H x}
\geq \frac{\phi(G, H)}{2} .
\]
\end{lemma}

\begin{proof}

We begin with a few algebraic identities:

Note that $2x_u^2+2x_v^2-(x_u - x_v)^2=(x_u+x_v)^2 \geq 0$
gives:
\[
\left(x_u - x_v\right)^2 \leq 2x_u ^2 + 2x_v^2.
\]

Also, suppose $uv \in E_H^{same}$ and without loss of generality that $|x_u| \geq |x_v|$.
Then letting $y = |x_u| - |x_v|$, we get:
\begin{eqnarray*}
|x_u^2-x_v^2| & = & \left(\left|x_v\right| + y\right) ^2 - \left|x_v\right|^2 \\
& = & y^2 + y |x_v|\\
& \geq & y^2 = \left(x_u - x_v\right)^2.
\end{eqnarray*}
The last equality follows because $x_u$ and $x_v$ have the same sign.

We then use the above inequalities to decompose the $x^T L_H x$ term.
\begin{eqnarray}
 x^T L_H   & = & \sum_{uv \in E_H^{same}}w_H(u,v)\left(x_u -x_v\right)^2
+ \sum_{uv \in E_H^{dif}}w_H(u, v)\left(x_u - x_v\right)^2 \nonumber \\
& \leq & \sum_{uv \in E_H^{same}} w_H(u, v) \left(x_u -x_v\right)^2
+ \sum_{uv \in E_H^{dif}} w_H(u, v) \left(2x_u^2 + 2x_v^2\right) \nonumber \\
& \leq & 2 \left( \sum_{uv \in E_H^{same}}  w_H(u, v) \left(x_u -x_v\right)^2
+ \sum_{uv \in E_H^{dif}} w_H(u, v) \left(x_u^2 + x_v^2\right) \right) \nonumber \\
& \leq & 2\left( \sum_{uv \in E_H^{same}}  w_H(u, v) \left|x_u^2 - x_v^2\right|
+ \sum_{uv \in E_H^{dif}} w_H(u, v) \left(x_u^2 + x_v^2\right) \right).
\end{eqnarray}

We can now decompose the summation further into parts for
$V^{-}$ and $V^{+}$:
\begin{align*}
& \sum_{uv \in E_G^{same}}w_G\left(u,v\right)\left|x_u^2-x_v^2\right| +
\sum_{uv \in E_G^{dif}}w_G\left(u, v\right) \left(x_u^2 + x_v^2\right) \\
= &\sum_{u \in V^{-}, v \in V^{-}}w_G\left(u,v\right)\left|x_u^2-x_v^2\right|
+ \sum_{u \in V^{-}, v \in V^{+}}w_G\left(u,v\right) x_u^2 \\
& + \sum_{u \in V^{+}, v \in V^{+}}w_G\left(u,v\right)\left|x_u^2-x_v^2\right|
+ \sum_{u \in V^{-}, v \in V^{+}}w_G\left(u,v\right) x_u^2.
\end{align*}

Doing the same for $ \sum_{uv \in E_H^{same}}  w_H(u, v) |x_u^2 - x_v^2|
+ \sum_{uv \in E_H^{dif}} w_H(u, v) (x_u^2 + x_v^2)$ we get:
\begin{align*}
& \frac{ \sum_{uv \in E_G^{same}}w_G(u,v)\left|x_u^2-x_v^2\right| +
\sum_{uv \in E_G^{dif}}w_G(u, v) \left(x_u^2 + x_v^2\right) }{x^T L_H x} \\
\geq & \min \left\{
\frac{\sum_{u \in V^{-}, v \in V^{-}}w_G(u,v)\left|x_u^2-x_v^2\right|
+ \sum_{u \in V^{-}, v \in V^{+}}w_G(u,v) x_u^2}
{\sum_{u \in V^{-}, v \in V^{-}}w_H(u,v)\left|x_u^2-x_v^2\right|
+ \sum_{u \in V^{-}, v \in V^{+}}w_H(u,v) x_u^2},\right.\\
&\qquad \left.\frac{\sum_{u \in V^{+}, v \in V^{+}}w_G(u,v)\left|x_u^2-x_v^2\right|
+ \sum_{u \in V^{-}, v \in V^{+}}w_G(u,v) x_v^2}
{\sum_{u \in V^{+}, v \in V^{+}}w_H(u,v)\left|x_u^2-x_v^2\right|
+ \sum_{u \in V^{-}, v \in V^{+}}w_H(u,v) x_v^2}
\right\}.
\end{align*}
The inequality comes from applying of Lemma \ref{th:splitting}.

By symmetry in $V^{-}$ and $V^{+}$, it suffices to show that
\begin{equation} \label{eq:toprove}
 \frac{\sum_{u\in V^{-}, v\in V^{-}}w_G\left(u, v\right)\left|x_u^2 -x_v^2\right| +
 \sum_{u\in V^{-}, v\in V^{+} }w_G(u, v) x_u^2}{
 \sum_{u\in V^{-}, v\in V^{-}}w_G\left(u, v\right)\left|x_u^2-x_v^2\right|
+ \sum_{u \in V^{-}, v \in V^{+}}w_G\left(u, v\right) x_u^2 }
\geq \phi(G, H).  \\
\end{equation}

We sort the $x_u$ in increasing order of $|x_u|$ into
such that $x_{u_1} \geq \ldots  \geq  x_{u_k}$, and let $S_k=\{x_{u_1},\ldots,x_{u_k}\}$. We have
\begin{align*}
& \sum_{u\in V^{-}, v\in V^{-}}w_G(u, v)\left|x_u^2 -x_v^2\right| +
 \sum_{u\in V^{-}, v\in V^{+} }w_G(u, v) x_u^2
=  \sum_{i=1 \dots k} \left(x_{u_i}^2 - x_{u_{i-1}}^2\right) cap_G\left(S_k, \bar{S_k}\right),
\end{align*}
and
\begin{align*}
& \sum_{u\in V^{-}, v\in V^{-}}w_H(u, v)\left|x_u^2 -x_v^2\right| +
 \sum_{u\in V^{-}, v\in V^{+} }w_H(u, v) x_u^2
 = \sum_{i=1 \dots k} \left(x_{u_i}^2 - x_{u_{i-1}}^2\right) cap_H\left(S_k, \bar{S_k}\right).
\end{align*}

Applying Lemma \ref{th:splitting} we have
\[
 \frac{\sum_{u\in V^{-}, v\in V^{-}}w_G(u, v)|x_u^2 -x_v^2| +
 \sum_{u\in V^{-}, v\in V^{+} }w_G(u, v) x_u^2}
 { \sum_{u\in V^{-}, v\in V^{-}}w_G\left(u, v\right)\left|x_u^2-x_v^2\right|
+ \sum_{u \in V^{-}, v \in V^{+}}w_G\left(u, v\right) x_u^2 }
\geq \min_k \frac{cap_H\left(S_G, \bar{S_i}\right)}{cap_H\left(S_i, \bar{S_i}\right)}
\geq \phi(G, H),
\]
where the second inequality is by definition of $\phi(G,H)$. This proves equation \ref{eq:toprove} and the Lemma follows.
\end{proof}

We now proceed with the proof of the main Theorem.

\begin{proof}

We have
\begin{eqnarray}
{x^TL_Gx} & = & \sum_{uv \in E_G}w_G(u,v)(x_u-x_v)^2 \nonumber\\
& = & \sum_{uv \in E_G^{same}}w_G(u,v)(x_u-x_v)^2
+ \sum_{uv \in E_G^{dif}}w_G(u,v)(x_u-x_v)^2 \nonumber \\
& \geq & \sum_{uv \in E_G^{same}}w_G(u,v)(x_u-x_v)^2
+ \sum_{uv \in E_G^{dif} }w_G(u,v) (x_u^2 + x_v^2). \nonumber \\
\end{eqnarray}
The last inequality follows by $x_ux_v \leq 0$ as $x_u \leq 0$ for all
$u \in V^{-}$ and $x_v \geq 0$ for all $v \in V^{+}$.

We multiply both sides of the inequality by
$$\sum_{uv \in E_G^{same}} w_G(u,v) (x_u + x_v)^2
 + \sum_{uv \in E_G^{dif}} w_G(u, v) (x_u^2 + x_v^2).$$

We have
\begin{eqnarray*}
&\left( \sum_{uv \in E_G^{same}}w_G(u,v)(x_u-x_v)^2
+ \sum_{uv \in E_G^{dif} }w_G(u,v) (x_u^2 + x_v^2)  \right)\\
& \cdot \left( \sum_{uv \in E_G^{same}} w_G(u,v) (x_u + x_v)^2
 + \sum_{uv \in E_G^{dif}} w_G(u, v) (x_u^2 + x_v^2) \right) \\
\geq & \left(  \sum_{uv \in E_G^{same}} |x_u - x_v| |x_u + x_v|
+ \sum_{uv \in E_G^{dif}} w_G(u, v) (x_u^2 + x_v^2) \right)^2\\
= & \left(  \sum_{uv \in E_G^{same}} |x_u^2 - x_v^2|
+ \sum_{uv \in E_G^{dif}} w_G(u, v) (x_u^2 + x_v^2) \right)^2.
\end{eqnarray*}

Furthermore, notice that $(x_u + x_v)^2 \leq 2x_u ^2 + 2x_v ^2$ since
$2x_u^2 + 2x_v^2 - (x_u + x_v)^2 = (x_u - x_v)^2 \geq 0$.
So, we have
\begin{align*}
& \sum_{uv \in E_G^{same}} w_G(u,v) (x_u + x_v)^2
 + \sum_{uv \in E_G^{dif}} w_G(u, v) (x_u^2 + x_v^2) \\
\leq & 2 \left( \sum_{uv \in E_G^{same}} w_G(u,v) (x_u^2 + x_v^2)
 + \sum_{uv \in E_G^{dif}} w_G(u, v) (x_u^2 + x_v^2) \right)\\
 & = 2 x^T Dx \leq 4 x^T L_{D_G} x,
\end{align*}
where $D$ is the diagonal of $L_G$ and the last inequality
comes from Lemma~\ref{th:DvsDG}. Combining the last two inequalities we get:
\begin{eqnarray*}
\frac{x^TL_Gx}{x^TL_Hx} \geq & \frac{1}{2}
\cdot \left( \frac{\sum_{uv \in E_G^{same}} \left|x_u^2 - x_v^2\right|
+ \sum_{uv \in E_G^{dif}} w_G(u, v) \left(x_u^2 + x_v^2\right)}{x^T L_H x} \right)\\
& \cdot \left( \frac{\sum_{uv \in E_G^{same}} \left|x_u^2 - x_v^2\right|
+ \sum_{uv \in E_G^{dif}} w_G(u, v) \left(x_u^2 + x_v^2\right)}{x^T L_{D_G} x} \right).
\end{eqnarray*}

By Lemma \ref{lem:abssqr}, we have that the first factor is bounded by
$\frac{1}{2} \phi(G, H)$  and the second factor bounded by
$\frac{1}{2} \phi(G, D_G)$. Hence we get
\begin{align}
\frac{x^TL_Gx}{x^TL_Hx}
& \geq \frac{1}{4} \phi(G, H) \phi(G,{D_G}).
\end{align}
\end{proof}

\section{Additional Experiments} \label{sec:additional}

{\bf PACM graph.} We again consider the (very) noisy ensemble \textit{NoisyKnn}($n=436, k_g=30,l_g=15$). Figure \ref{fig:PACM_clusterings} shows a random instance of the clustering returned by each of the methods, with 125 constraints. Figure \ref{fig:PACM_curves} shows the accuracy and running times of all three methods on this example.
Again, our approach returns superior results when compared to \textbf{CSP}, and it is somewhat better than \textbf{COSf}. In this example, our running time is larger than that of both \textbf{COSf} and \textbf{CSP}, which is due to the small size of the problem ($n=426$). For such small problems a full eigenvalue decomposition is faster due to its better utilization of the FPU, as well as some overheads of the iterative method (e.g. preconditioning). In principle we can use the full eigenvalue decomposition to speed-up our algorithm for these smaller problems and at least match the performance of~\textbf{CSP}. However the running times are already very small.

\begin{figure}[h]
\begin{center}
\includegraphics[width=0.30\columnwidth]{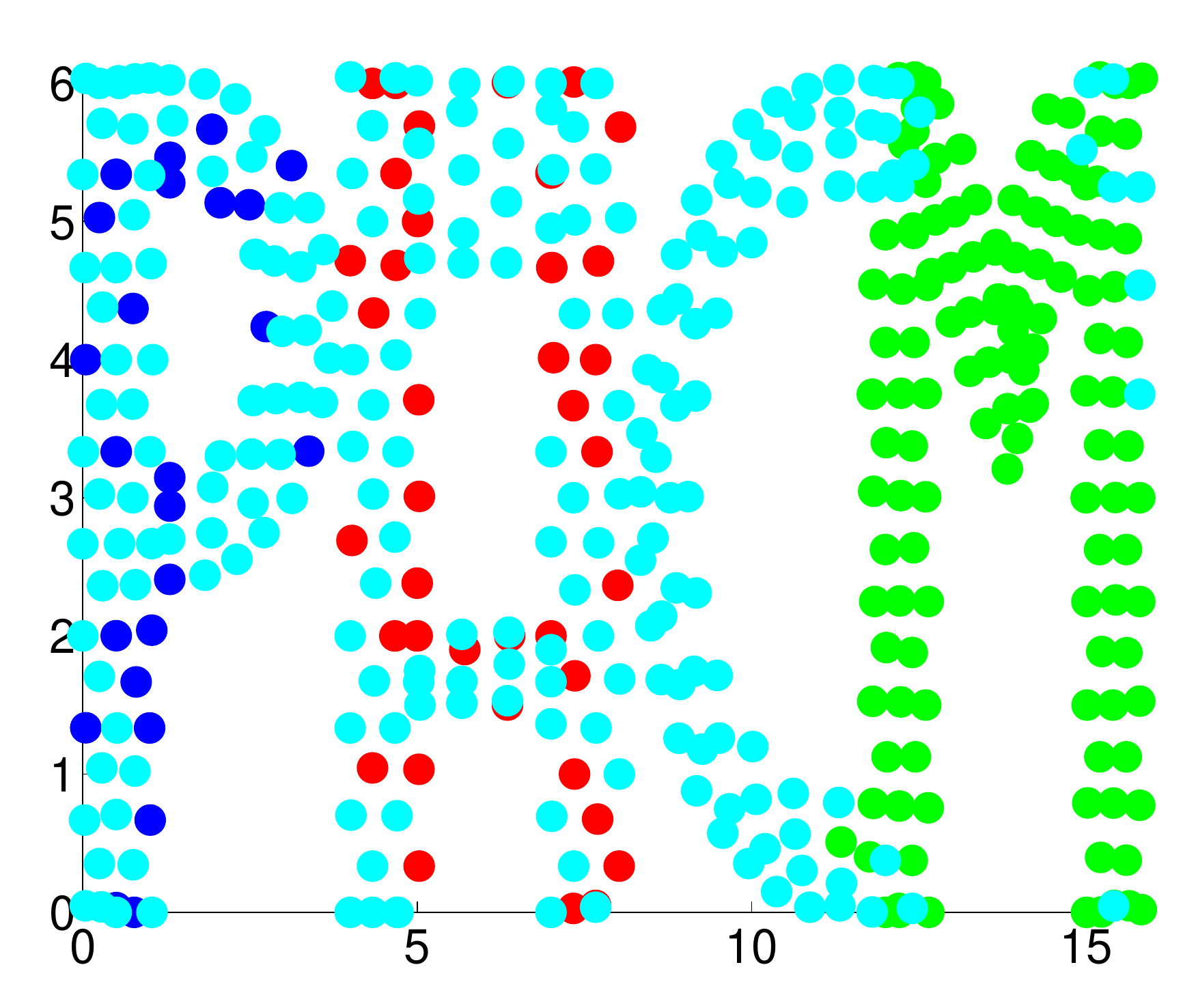}
\includegraphics[width=0.30\columnwidth]{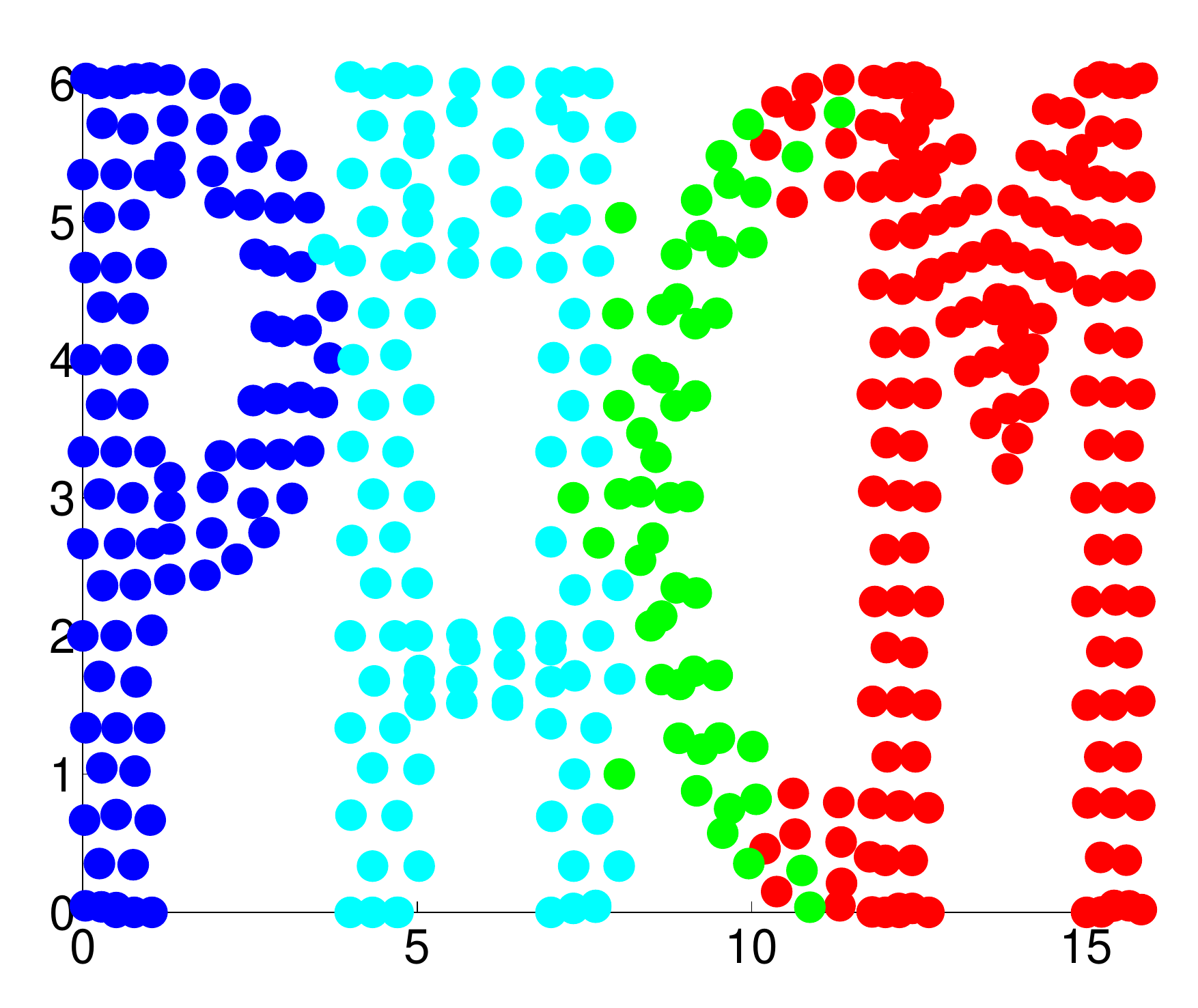}
\includegraphics[width=0.30\columnwidth]{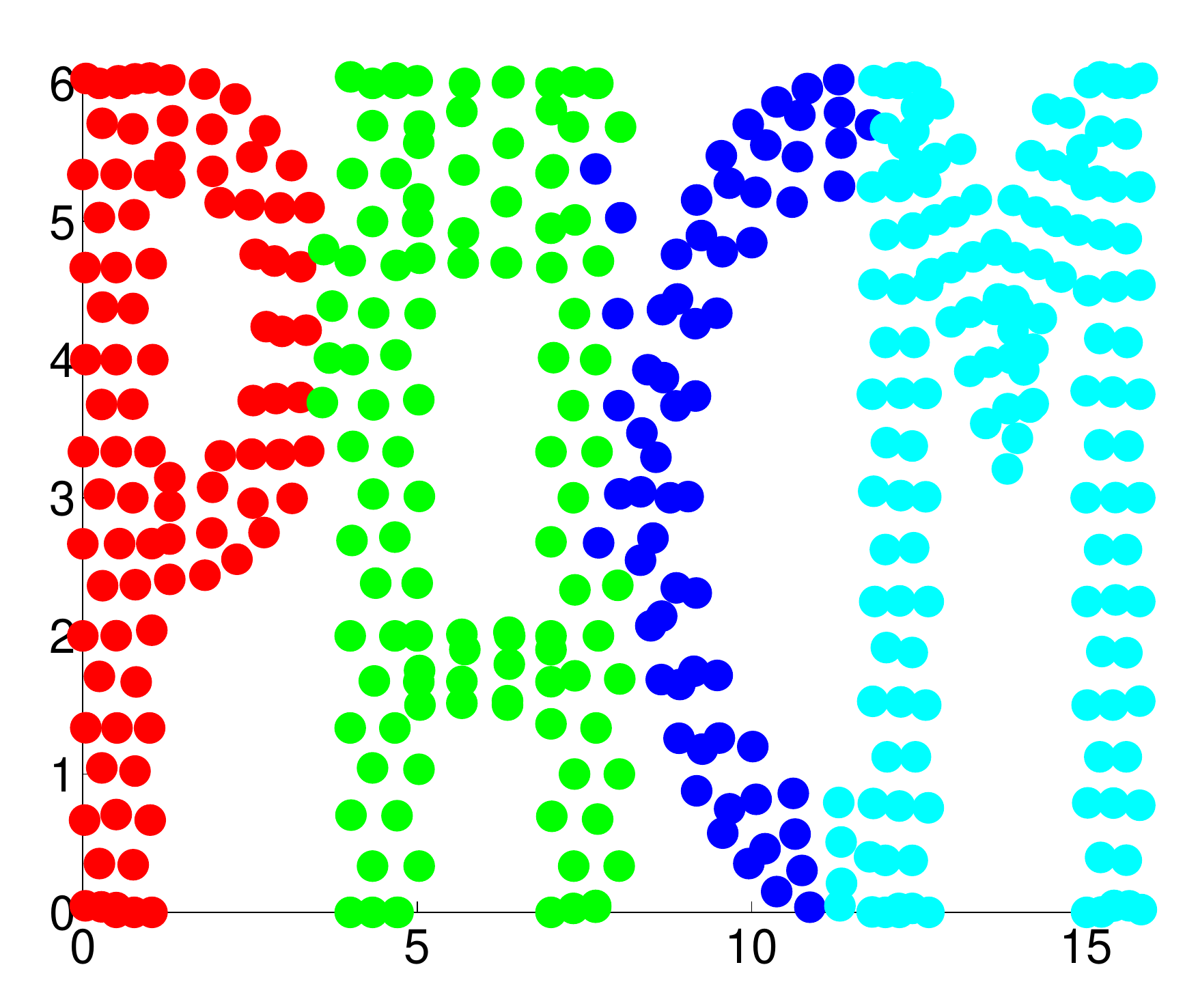}
\end{center}
\caption{Top: Segmentation for a random instance of the PACM data set with $125$ labels produced by \textbf{CSP} (left), \textbf{COSf} (middle) and \textbf{FAST-GE} (right)}
\label{fig:PACM_clusterings}
\end{figure}

\begin{figure}[h]
\begin{center}
\includegraphics[width=.48\columnwidth]{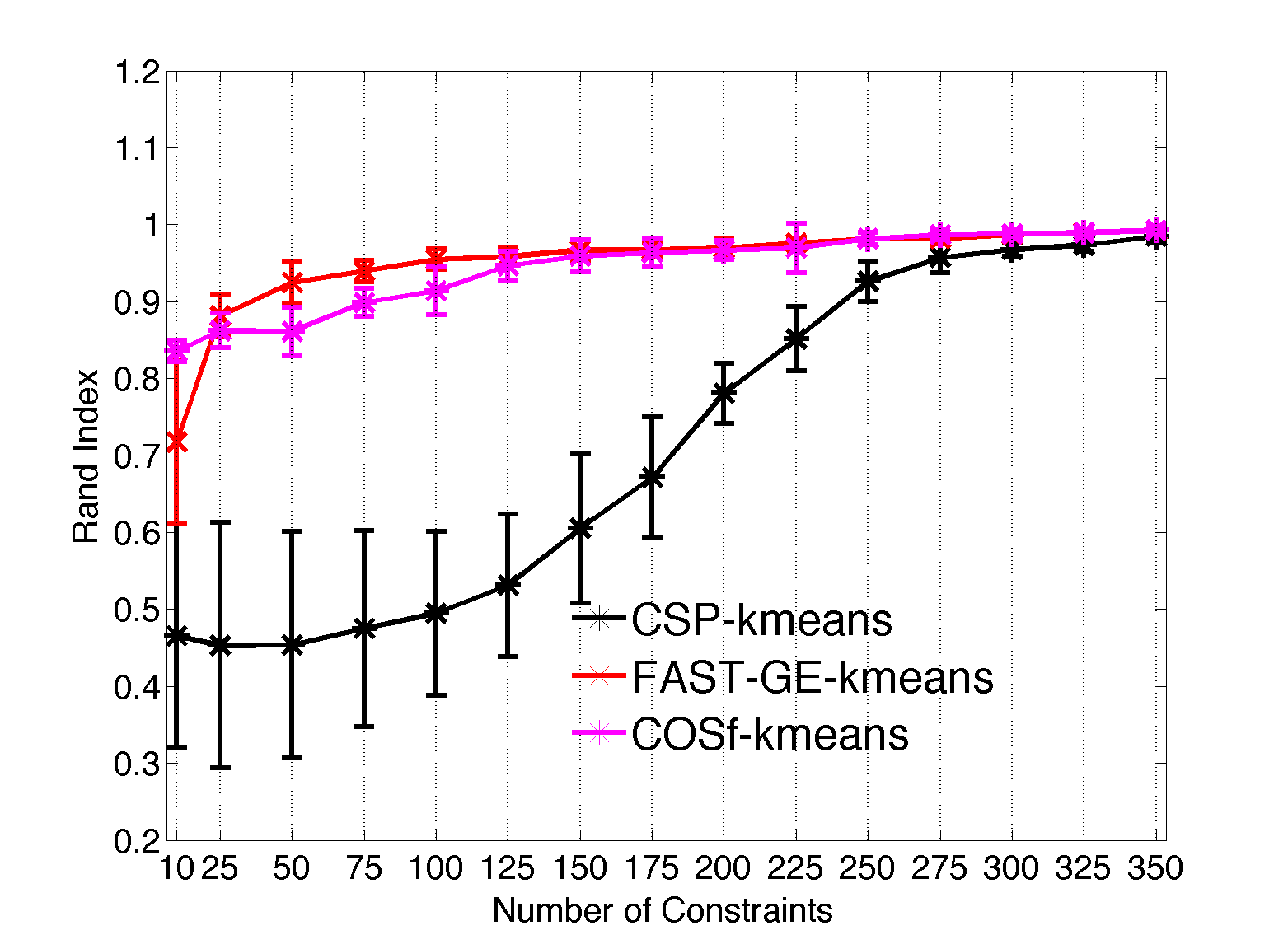}
\includegraphics[width=.48\columnwidth]{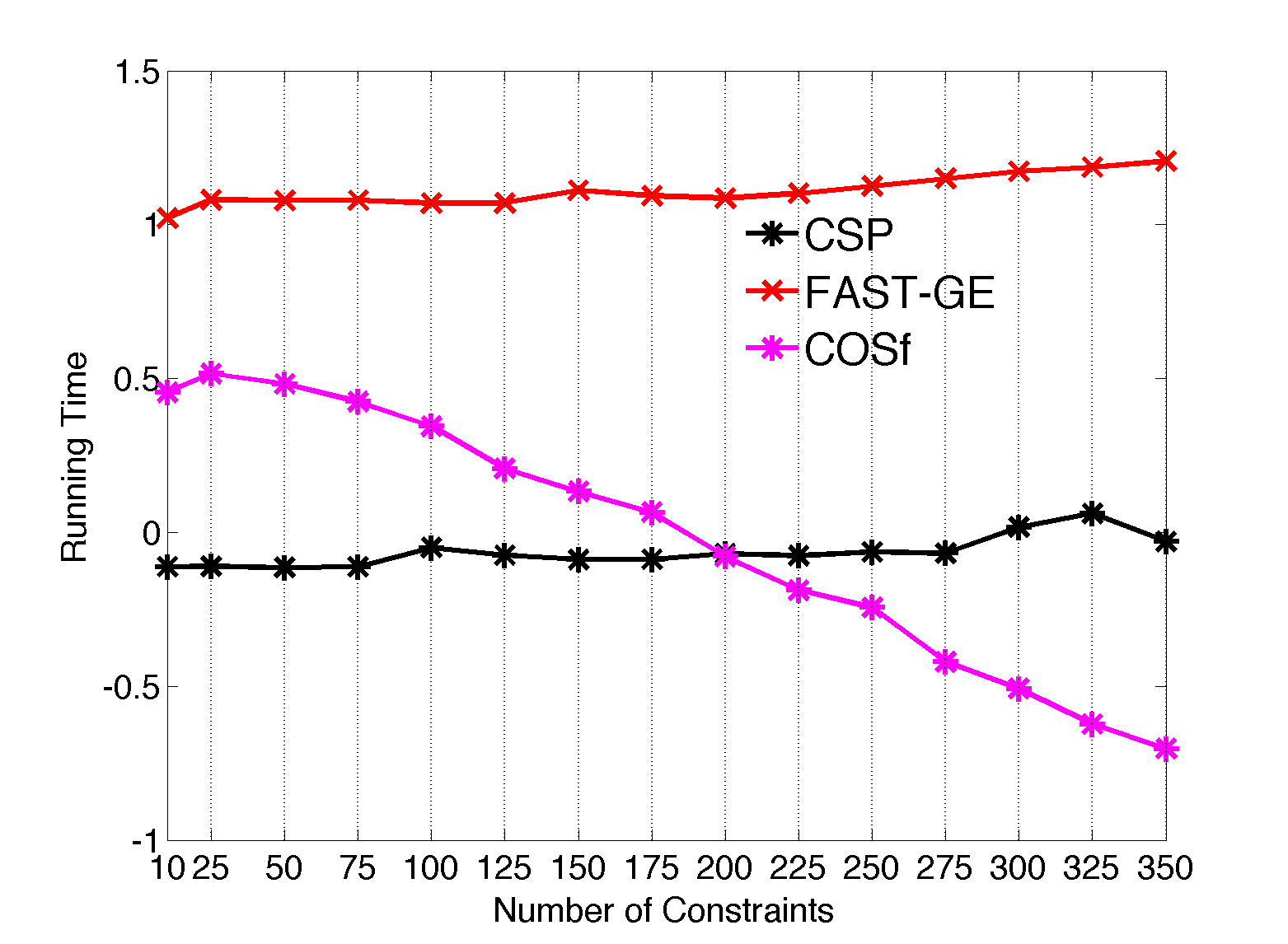}
\end{center}
\caption{Leftmost plots illustrate the accuracy and running times for the Four-Moons data set, where the underlying graph given by the model NoisyKnn($n=1500, k=30,l=15$), for varying  number of constraints. 
The rightmost two plots show similar statistics for the PACM data set, with the noise model given by NoisyKnn($n=426$, $k=30, l=15$). We average all results over 20 runs.}
\label{fig:PACM_curves}
\end{figure}